\documentclass{birkmult}

\usepackage{color,amsmath, amssymb}
\usepackage{pgf}
\usepackage{mathtools}
\usepackage{framed,graphicx,xcolor}
\usepackage{accents}
\newcommand{\dbtilde}[1]{\accentset{\approx}{#1}}
\usepackage{tikz}
\usepackage{amssymb}
\usepackage{cite}

 \newtheorem{thm}{Theorem}[section]
 \newtheorem{cor}[thm]{Corollary}
 
 \newtheorem{prop}[thm]{Proposition}
 \theoremstyle{definition}
 \newtheorem{defn}[thm]{Definition}

 \newtheorem{rem}[thm]{Remark}

 \numberwithin{equation}{section}
\newtheorem{exercise}{Exercise}
\newtheorem{example}[thm]{Example}
\newcommand{\ddz}{\frac{{\rm d}}{{\rm d}z}}
\renewcommand{\a}{\alpha}

\newcommand{\p}{Painlev\'{e}}
\newcommand{\hW}{{\widehat{W}}}

\newcommand{\tH}{\tilde{H}}
\newcommand{\hT}{\hat{T}}

\newcommand{\cc}[1]{c_{#1}}

\newcommand{\lp}{\left(}
\newcommand{\rp}{\right)}

\DeclareMathOperator{\rL}{\operatorname{L}}
\DeclareMathOperator{\Wr}{\operatorname{Wr}}

\newcommand{\N}{\mathbb{N}}

\newcommand{\Z}{\mathbb{Z}}
\newcommand{\C}{\mathbb{C}}
\newcommand{\R}{\mathbb{R}}
\newcommand{\bbeta}{{\boldsymbol{\beta}}}
\newcommand{\bmu}{{\boldsymbol{\mu}}}
\newcommand{\bpi}{{\boldsymbol{\pi}}}

\newcommand{\supth}{{}^{\rm{th}}}

\newcommand{\PIV}{$\mathrm{P}_{\mathrm{IV}}\,\,$}
\newcommand{\cM}{\mathcal{M}}
\newcommand{\cZ}{\mathcal{Z}}

\newcommand{\hH}{{\widehat{H}}}

\renewcommand{\th}{\tilde{H}}

\newcommand{\Pfour}{{\rm P_{IV}}}
\newcommand{\Pfive}{{\rm P_{V}}}
\newcommand{\Pone}{{\rm P_{I}}}

\renewcommand{\boxdot}{{\ \clap{\raise0.25ex\hbox{$\bullet$}}\clap{$\square$}\ }}
\newcommand{\emptybox}{\hbox{$\square$}}

\definecolor{shadecolor}{gray}{0.9}

\begin{document}
%
%
%
%
%
%
%
%
%

\title[Exceptional orthogonal polynomials and Painlev\'e equations]{Lectures on \mbox{exceptional orthogonal polynomials} and rational solutions to Painlev\'e equations.}
\author{David G\'omez-Ullate}

\address{Escuela Superior de Ingenier\'ia, Universidad de C\'adiz, 11519 Puerto Real, Spain.}
\address{Departamento de F\'isica Te\'orica, Universidad Complutense de
  Madrid, 28040 Madrid, Spain.}

\email{david.gomezullate@uca.es}

\author{Robert Milson}
\address{Department of Mathematics and Statistics, Dalhousie University,
  Halifax, NS, B3H 3J5, Canada.}
\email{rmilson@dal.ca}


\date{Feb 20, 2019}

\begin{abstract}
These are the lecture notes for a course on exceptional polynomials taught at the \textit{AIMS-Volkswagen Stiftung Workshop on Introduction to Orthogonal Polynomials and Applications} that took place in Douala (Cameroon) from October 5-12, 2018. They summarize the basic results and construction of exceptional poynomials, developed over the past ten years. In addition, some new results are presented on the construction of rational solutions to Painlev\'e equation $\Pfour$ and its higher order generalizations that belong to the $A_{2n}^{(1)}$-Painlev\'e hierarchy. The construction is based on dressing chains of Schr\"odinger operators with potentials that are rational extensions of the harmonic oscillator. Some of the material presented here (Sturm-Liouville operators, classical orthogonal polynomials, Darboux-Crum transformations, etc.) are classical and can be found in many textbooks, while some results (genus, interlacing and cyclic  Maya diagrams) are new and presented for the first time in this set of lecture notes.
\end{abstract}

\maketitle
\setcounter{tocdepth}{1}
\tableofcontents

\section{Introduction}\label{sec:intro}
   
The past 10 years have witnessed an intense activity of several research groups around the concept of exceptional orthogonal polynomials. Although some isolated examples in the physics literature existed before \cite{dubov}, the systematic study of exceptional polynomials started in 2009, with the publication of two papers \cite{gomez2009extended,gomez2010extension}. The original approach to the problem was via a complete classification of exceptional operators by increasing codimension, which proved to be computationally untractable (and moreover, much later it was shown that codimension is not a very well defined concept).  The term \textit{exceptional} was originally intended to evoque very rare, almost exotic cases, as for low codimension exceptional families are almost unique \cite{gomez2010extension}. Yet, shortly after the publication of these results, Quesne showed \cite{Quesne2008} that the exceptional families in \cite{gomez2009extended} could be obtained by Darboux transformations, and Odake and Sasaki showed the way to generalize these examples to arbitrary codimension \cite{Odake2009a,Odake2010}. New families emerged later associated to multiple Darboux transformations \cite{Gomez-Ullate2012,odake2011exactly}, and nowadays it is clear that exceptional polynomials are certainly not rare, as we are starting to understand the whole theory behind their construction and classification.

The role of Darboux transformations in the construction of exceptional polynomial families is an essential ingredient. It was conjectured in \cite{gomez2013conjecture} that every exceptional polynomial system can be obtained from a classical one via a sequence of Darboux transformations, which has been recently proved in \cite{garcia2016bochner}. Explaining these results lies beyond the scope of these lectures, and the interested reader is advised to read \cite{garcia2016bochner} for an updated account on the structure theorems underlying the theory of exceptional polynomials and operators. We will limit ourselves in the following pages to introduce the main ideas and constructions, as a sort of primer to the subject.

For those interested in gaining deeper knowledge in the properties of exceptional orthogonal polynomials, there are a large number of references in the bibliography section that cover the main results published in the past 10 years, by authors like Dur\'an \cite{Duran2014a,Duran2014,Duran2015Jacobi,Duran2015a,Duran2015b}, Sasaki and Odake \cite{Odake2009a,Odake2010,odake2011exactly,odake2013extensions,odake2013krein}, Marquette and Quesne \cite{marquette2013one,marquette2013two,marquette2016}, Kuijlaars and Bonneux \cite{Kuijlaars2015,Bonneux2018}, etc. that cover aspects like recurrence relations, symmetries, asymptotics, admissibility and regularity of the weights, properties of their zeros and electrostatic interpretation, and applications in solvable quantum mechanical models, among others.

The connection between sequences of Darboux transformations and Painlev\'e type equations has been known for more than 20 years, since the works of Adler \cite{adler1994nonlinear}, and  Veselov and Shabat ,\cite{veselov1993dressing}. However, the russian school of integrable systems was more concerned with uncovering relations between different structures rather than providing complete classifications of solutions to Painlev\'e equations. The japanese school pionereed by Sato developed a scheme to understand integrable equations as reductions from the KP equations. Noumi and Yamada \cite{noumi2004painleve}, and their collaborators developed the geometric theory of Painlev\'e equations, by studying the group of B\"acklund transformations that map solutions to solutions (albeit for different values of the parameters). Using this transformations to \textit{dress} some very simple seed solutions they managed to build and classify large classes of rational solutions to $\Pfour$ and $\Pfive$, and to extend this symmetry approach to higher order equations, that now bear their name. It was later realised that determinantal representations of these rational solutions exist \cite{kajiwara1996determinant,kajiwara1998determinant} and that they involve classical orthogonal polynomial entries. For an updated account of the relation between orthogonal polynomials and Painlev\'e equations, the reader is advised to read the recent book by van Assche \cite{refWVAbook}.

Our aim is to merge these two approaches: the strength of the Darboux dressing chain formulation with a convenient representation and indexing to describe the whole set of rational solutions to $\Pfour$ and its higher order generalizatins belonging to the $A_{2N}$-Painlev\'e hierarchy. This is achieved by indexing iterated Darboux transformations with Maya diagrams, originally introduced by Sato, and exploring conditions that ensure cyclicity after an odd number of steps. We tackle this problem by introducing the concepts of genus and interlacing of Maya diagrams, which allow us to classify and describe cyclic Maya diagrams. For every such cycle, we show how to build a rational solution to the $A_{2N}$-Painlev\'e system, by a suitable choice of Wronskian determinants whose entries are Hermite polynomials. This approach generalizes the solutions for $\Pfour$ ($A_{2}$-Painlev\'e) known in the literature as Okamoto and generalized Hermite polynomials. We illustrate the construction by providing the complete set of rational solutions to $A_{4}$-Painlev\'e, the next system in the hierarchy.

\section{Darboux transformations}\label{sec:Darboux}

In this section we describe Darboux transformations on Schr\"odinger operators and their iterations at a purely formal level (i.e. with no interest on the spectral properties).

Let $L=-D_{xx}+U(x)$ be a Schr\"odinger operator, and $\varphi=\varphi(x)$ a formal eigenfunction of $L$ with eigenvalue $\lambda$, so that
\[ L[\varphi]=-\varphi''+U\varphi=\lambda\varphi\]
Note that we are not assuming any condition of $\varphi$ at this stage, we do not care at this formal level whether  $\varphi$ is square integrable or not.
The function $\varphi$ is usually called the \textit{seed function} for the transformation, and $\lambda$ the \textit{factorization energy}.
For every choice of $\varphi$ and $\lambda$, we can factorize $L$ in the following manner
\[L-\lambda=(D_x+w)(-D_x+w)=BA,\qquad w=(\log \varphi)'\]
where $B=D_x+w$ and $A=-D_x+w$ are first order differential operators.
The Darboux transform of $L$, that we call $\tilde L$, is defined by commuting the two factors:
\[ \tilde L-\lambda= AB= (-D_x+w)(D_x+w)\]
Expanding the two factors, we can find the relation between $U$ and its transform $\tilde U$:
\begin{eqnarray}
U=w^2+w'-\lambda,\quad \tilde U=w^2-w'-\lambda\Rightarrow \tilde U=U-2w'
\end{eqnarray}
or in terms of the seed function we have
\[\tilde U=U-2(\log\varphi)''\]
Note that $\ker A=\langle \varphi\rangle$, i.e. $A[\varphi]=-\varphi'+w\varphi=0$,and also that  $\ker B=\left\langle \frac{1}{ \varphi}\right\rangle$, i.e. 
$B\left[\frac{1}{\varphi}\right]=0$.
The main reason to introduce this transformation is that we have the following intertwining relations between $L$ and $\tilde L$:
\begin{equation}\label{eq:intertwining}
 LB=B\tilde L,\qquad AL=\tilde LA 
 \end{equation}
These relations mean that we can connect the eigenfunctions of $L$ and $\tilde L$. 
\begin{shaded}
\begin{exercise}
Show that if $\psi$ is an eigenfunction of $L$ with eigenvalue $E$, then $\tilde \psi=A[\psi]$ is an eigenfunction of $\tilde L$ with the same eigenvalue. Likewise, if $\tilde\varphi$ is an eigenfunction of $\tilde L$ with eigenvalue $\mu$, then $\varphi=B[\tilde \varphi]$ is an eigenfunction of $L$ with the same eigenvalue.
\end{exercise}
\end{shaded}
By hypothesis, we have that $L[\psi]=E\psi$. Let $\tilde \psi=A[\psi]$. We see that
\[ \tilde L[\tilde \psi]=\tilde{L}A[\psi]=AL[\psi]=A[E\psi]=\mu A[\psi]=E\tilde\psi
\]
The converse transformation is proved in a similar way.

Note that if we try to apply the Darboux transformation $A$ on $\varphi$ we do not get any eigenfunction of $\tilde L$, because $A[\varphi]=0$. However, the reciprocal of $\varphi$ is a new eigenfunction of $\tilde L$, with eigenvalue $\lambda$, as
\[ B\left[\frac{1}{\varphi}\right]=0\Rightarrow \tilde L\left[\frac{1}{\varphi}\right] = (AB+\lambda) \left[\frac{1}{\varphi}\right]=\lambda \left(\frac{1}{\varphi}\right)\]

\subsection{Exact solvability by polynomials}\label{sec:ESP}

The above transformation is purely formal and its main purpose is to connect the eigenfunctions and eigenvalues of two different Schr\"odinger operators $L$ and $\tilde L$. Now, this would be of very little purpose if we cannot say anything about the spectrum and eigenfunctions of at least one of the operators. Typically, we use Darboux transformations to generate new solvable operators from ones that we know to be solvable. But what do we mean by \textit{solvable} ?

In general, \textit{solvable} means that we can describe the spectrum and eigenfunctions in a more or less explicit form, and in terms of known functions. It is still not clear what a \textit{known function} is\dots so we'd rather narrow down the definition and define exact solvability by polynomials in the following manner

\begin{defn}\label{def:ESP}
  A Schr\"odinger operator
  \begin{equation}\label{eq:Hgen}
  L=-D_{xx} +U(x)
 \end{equation}  
 is said to be \textit{exactly solvable by polynomials} if there exist functions
 $\mu(x),z(x)$ such that \emph{for all but finitely many} $k\in
 \mathbb N$, $L$ has eigenfunctions (in the ${\mathrm L}^2$ sense) of the form 
  \[ \psi_k(x) = \mu(x) y_k(z(x)) \]
where $y_k(z)$ is degree $k$ polynomial in $z$.
\end{defn}

This definition captures many of the Schr\"odinger operators that we know to be exactly solvable: those in which the eigenfunctions have a common prefactor $\mu(x)$ times a polynomial in a suitable variable $z(x)$. The prefactor  $\mu(x)$ is responsible for ensuring the right asymptotic behaviour at the endpoints for all bound states, while the polynomials $y_k$ represent a modulation that describes the excited states.

From the purpose of orthogonal polynomials, this kind of Schr\"odinger operators are directly related to classical orthogonal polynomials, since polynomials $y_k$ are automatically orthogonal if $L$ is a self-adjoint operator, i.e. with appropriate regularity and boundary conditions.

More specifically, classical orthogonal polynomials are related to the following Schr\"odinger operators:
\begin{enumerate}
\item Hermite polynomials to the harmonic oscillator
\item Laguerre polynomials to the isotonic oscillator
\item Jacobi polynomials to the Darboux-P\"oschl-Teller potential.
\end{enumerate}

We would like to apply Darboux transformations to these three families of Schr\"odinger operators that are exactly solvable by polynomials, in order to generate new operators, but we would like these new operators to also be exactly solvable by polynomials. This means that we have to impose certain restrictions on the type of seed functions of $L$ that we are free to choose for the Darboux transformations. In general, the class of \textit{rational Darboux transformation} is the subset of all possible Darboux transformations that preserve exact solvability by polynomials. Fortunately, there is a simple way to characterize seed functions for this subclass, which we describe below. But before we do so, let us introduce some jargon between differential operators.

\subsection{Schr\"odinger and algebraic operators}\label{sec:ratDarb}

If we are dealing with Schr\"odinger operators that are exactly solvable by polynomials, there are two operators that we will work with: on one hand, we have the Schr\"odinger operator $L=-D_{xx}+U(x)$, on the other hand, we have the algebraic operator $T=p(z)D_{zz}+q(z)D_z+r(z)$ which is the one that has polynomial eigenfunctions
\[T[y_k]=p(z) y_k''+q(z) y_k'+r(z) y_k=\lambda_k y_k.\]
There is a connection between these two operators, which is not entirely bidirectional in general, but it is bidirectional if $L$ is exactly solvable by polynomials.

\begin{prop}
Every second order linear differential operator $T=p(z)D_{zz}+q(z)D_z+r(z)$ can be transformed into a Schr\"odinger operator $L=-D_{xx}+U(x)$ by the following change of variables and similarity transformation:
\begin{eqnarray}\label{eq:xtoz}
x&=&-\int^z (-p)^{-1/2}\\
L&=&\mu\circ T\circ \left(\frac{1}{\mu}\right)\Big|_{z=z(x)},\qquad \mu=\exp\int^z\frac{q-\frac{p'}{2}}{2p}\label{eq:TtoL}
\end{eqnarray}
where $z(x)$ is defined by inverting \eqref{eq:xtoz}
\end{prop}

\begin{shaded}
\begin{exercise}
Prove that if $T=p(z)D_{zz}+q(z)D_z+r(z)$, then the operator $L$ defined by \eqref{eq:xtoz} and \eqref{eq:TtoL} is a Schr\"odinger operator $L=-D_{xx}+U(x)$, and find an expression for the potential $U$ in terms of $p,q$ and $r$.
\end{exercise}
\end{shaded}

Let us denote
\[ \mu(z)=\exp\int^z \sigma,\qquad \sigma=\frac{q-\frac{p'}{2}}{2p}. \]
This, by applying the product rule for composition of differential operators, we have
\begin{eqnarray*}
\mu\circ D_z\circ \left(\frac{1}{\mu}\right)&=&D_z+\left(\log\frac{1}{\mu} \right)'=D_z-\sigma\\
\mu\circ D_{zz}\circ \left(\frac{1}{\mu}\right)&=& D_{zz}+2 \left(\log\frac{1}{\mu} \right)'D_z+\mu\left(\frac{1}{\mu}\right)''=D_{zz}-2\sigma D_z-\sigma'+\sigma^2
\end{eqnarray*}
So collecting terms we have
\[ \mu\circ T\circ \left(\frac{1}{\mu}\right)=pD_{zz} +\frac{1}{2} p' D_z+r-q\sigma-p\sigma'+p\sigma^2 \]

Finally, the chain rule for differentiation leads to
\begin{align*}
\frac{d}{dz}&=\frac{dx}{dz}\frac{d}{dx}&&\Rightarrow \qquad D_z=-(-p)^{1/2}D_x\\
\frac{d^2}{dz^2}&=\left( \frac{dx}{dz}\right)^2\frac{d^2}{dx^2}+ \frac{d^2x}{dz^2}\frac{d}{dx}&&\Rightarrow\qquad D_{zz}=-\frac{1}{p}D_{xx}+\frac{1}{2}(-p)^{1/2}p'D_x
\end{align*}
which inserted into the previous equation becomes
\[ \mu\circ T\circ \left(\frac{1}{\mu}\right) =-D_{xx}+r+p\sigma^2-p\sigma'-q\sigma \]
which is a Schr\"odinger operator. We identify the potential $U(x)$ to be
\begin{equation}
U(x)=r+p\sigma^2-p\sigma'-q\sigma\Big|_{z=z(x)},\qquad \sigma=\frac{q-\frac{p'}{2}}{2p}
\end{equation}

Note that we can always go from $T$ to $L$, but in general there is no prescribed way to go from $L$ to $T$. This means that  given a Schr\"odinger operator with some potential $L=-D_{xx}+U(x)$ it is a difficult question to know  whether we can perform a change of variables and conjugation by a factor $\mu$ as in \eqref{eq:xtoz}-\eqref{eq:TtoL} such that $T$ has polynomial eigenfunctions $y_k(z)$ (this is sometimes known as \textit{algebraizing} a Schr\"odinger operator). Otherwise speaking, given a potential $U$ it is hard to know whether it is exactly solvable by polynomials\footnote{Solving this question is equivalent to classifying exceptional polynomials and operators, a question that we shall mention below.}.

\subsection{Rational Darboux transformations}\label{sec:ratDarb}
Now we start from a given $L$ which we know to be exactly solvable by polynomials, and we would like to perform a Darboux transformation in such a way that $\tilde L$ is still exactly solvable by polynomials. What conditions must the seed function $\varphi$ satisfy for this to hold?

This question is not easy to answer in general. Let us look at one example.
\begin{example}\label{ex:ho}
Consider the harmonic oscillator $L=-D_{xx}+x^2$. One possible choice for seed functions for rational Darboux transformations comes from choosing $\varphi$ among the bound states of $L$, i.e.
\[\varphi_k={\rm e}^{-x^2/2} H_k(x),\qquad k\geq 0\]
where $H_n(x)$ is the $n$-th Hermite polynomial. We see that $L$ is exactly solvable by polynomials, with $z(x)=x$ and $\mu(x)={\rm e}^{-x^2/2}$. This implies that $p(z)=1$ from  \eqref{eq:xtoz}, and from \eqref{eq:TtoL} we see that $q(z)=2 z$, so that
$T=D_{zz}-2zD_z$
\end{example}
We thus have
\[
\begin{aligned}
&L=-D_{xx}+x^2,\qquad &&L[\varphi_k]=(2k+1) \varphi_k,\quad &k=0,1,2,\dots\\
 &T= D_{zz}-2 z D_z,\qquad &&T[H_k]= 2k H_k,\quad &k=0,1,2,\dots 
\end{aligned}
\]
But this is not the only possible choice. Note that the Schr\"odinger operator $L=-D_{xx}+x^2$ is not only invariant under the transformation $x\to-x$ but it only picks a sign when we perform the transformation $x\to{\rm i} x$, so there is another set of eigenfunctions
\[ \tilde \varphi_k= {\rm e}^{x^2/2} {\tilde H}_k(x),\quad k\geq 1  \]
where ${\tilde H}_k(x)={\rm i}^{-k}\,H_k({\rm i}x)$ is called the conjugate Hermite polynomial. Note that these eigenfunctions are obtained by exploiting a discrete symmetry of the equation, and their eigenvalues are negative:
\[ L[\tilde\varphi_k]=-(2k+1) \tilde\varphi_k,\quad k=1,2,\dots \]
Because the pre-factor is now a positive gaussian, the functions blow up at $\pm\infty$ and they are not square integrable (in the physics literature they are sometimes called \textit{virtual states}). But for the purposes of using them as seed functions for Darboux transformations, they are perfectly valid.
These two sets of eigenfunctions exhaust all possible seed functions for rational Darboux transformations of the harmonic oscillator. Rather than two families of functions, each of them indexed by natural numbers, it will be useful to consider them as one single family indexed by integers:

\begin{equation}\label{eq:seedHO}
\varphi_n(x)=\begin{cases}
{\rm e}^{-x^2/2} H_n(x),&\text{ if } n\geq 0\\
{\rm e}^{x^2/2} \tilde H_{-n-1}(x),&\text{ if } n< 0
\end{cases}
\end{equation}

\subsection{Iterated or Darboux-Crum transformations}\label{sec:Crum}

Now that we know how to apply Darboux transformations to pass from $L$ to $\tilde L$ there is no reason why we should stop there\dots we can apply the transformation once again, and in general as many times as we want. Let's do it once more. Suppose that $\varphi_1$ and $\varphi_2$ are two formal eigenfunctions of $L$:
\[L[\varphi_1]=\lambda_1 \varphi_1,\qquad L[\varphi_2]=\lambda_2 \varphi_2  \]
We first use seed function $\varphi_1$ to transform $L$ into $\tilde L$, so that
\[\tilde U = U- 2(\log\varphi_1)'' \]
and the eigenfunctions are related by
\[\tilde\psi =A[\psi]=-\psi'+\frac{\varphi_1'}{\varphi_1}\psi=\frac{\Wr[\varphi_1,\psi]}{\varphi_1}. \]
But now we observe that $\tilde\varphi_2$ given by
\[ \tilde\varphi_2 =A[\varphi_2]=\frac{\Wr[\varphi_1,\varphi_2]}{\varphi_1}\]
is a formal eigenfunction of $\tilde L$, so we can use it to Darboux transform $\tilde L$ into $\dbtilde L$, so
\[ \dbtilde U=\tilde U -2(\log\tilde\varphi_2)''=U-2(\log\Wr[\varphi_1,\varphi_2])''  \]
If $\psi$ are the eigenfunctions of $L$ and $\tilde\psi=\Wr[\varphi_1,\psi]/\varphi_1$ are the eigenfunctions of $\tilde L$, the eigenfunctions of $\dbtilde L$ are given by
\[\dbtilde\psi=A[\tilde\psi]=\frac{\Wr[\tilde\varphi_2,\tilde\psi]}{\tilde\varphi_2}=\frac{
\Wr \Big[ \varphi_1^{-1}\Wr[\varphi_1,\varphi_2],\varphi_1^{-1}\Wr[\varphi_1,\psi]  \Big]}{\varphi_1^{-1}\Wr[\varphi_1,\varphi_2]}=\frac{\Wr[\varphi_1,\varphi_2,\psi]}{\Wr[\varphi_1,\varphi_2]}
\]
where we have used two identities satisfied by Wronskian determinants, namely
\[\Wr[gf_1,\dots,gf_n] =g^n\Wr[f_1,\dots,f_n]\]
and
\[\Wr[f_1,\dots,f_n,g,h]=\frac{\Wr\big[\Wr[f_1,\dots,f_n,g],\Wr[f_1,\dots,f_{n},h]   \big]}{\Wr[f_1,\dots,f_n]}. \]

It is not hard to iterate this argument and prove by induction the following result, known as Darboux-Crum formula.
\begin{prop}
Let $\varphi_1,\dots,\varphi_n$ be a set of $n$ formal eigenfunctions of a Schr\"odinger operator $L$. We can perform an $n$-step Darboux transformation with these seed eigenfunctions, to obtain a chain of Schr\"odinger operators
\[L=L_0 \rightarrow L_1 \rightarrow \cdots \rightarrow L_n. \]
The Schr\"odinger operator of $L_n$ is given by
\[ L_n=-D_{xx}+ U_n=-D_{xx}+ U -2\left(\log\Wr[\varphi_1,\dots,\varphi_n]\right)''.  \]
If $\psi$ is a formal eigenfunction of $L$ with eigenvalue $E$, then
\begin{equation}\label{eq:Crum}
 \psi^{(n)} = \frac{\Wr[\varphi_1,\dots,\varphi_n,\psi]}{\Wr[\varphi_1,\dots\varphi_n]}
 \end{equation}
is a formal eigenfunction of $L_n$ with the same eigenvalue.
\end{prop}

\begin{example}
Coming back to the harmonic oscillator of Example~\ref{ex:ho}, we saw that seed functions for rational Darboux transformations are in one-to-one correspondence with the integers. 
If we want to perform a multi-step Darboux transformation, we need to fix a multi-index that specifies the set of seed functions to be used.
For instance, corresponding to the multi-index $(-3,-2,1,4)$ we would have, according to \eqref{eq:seedHO} the Darboux-Crum transformation acting on a function $\psi$ would be
\[\psi^{(4)}=\frac{ \Wr\left[{\rm e}^{x^2/2} \tH_2,{\rm e}^{x^2/2}\tH_1,{\rm e}^{-x^2/2} H_1,{\rm e}^{-x^2/2} H_4,\psi\right]}{\Wr\left[{\rm e}^{x^2/2} \tH_2,{\rm e}^{x^2/2}\tH_1,{\rm e}^{-x^2/2} H_1,{\rm e}^{-x^2/2} H_4\right] }
\]
In the following sections we will see how the polynomial part of these functions essentially defines exceptional Hermite polynomials, and how these Wronskians enjoy very particular symmetry properties that admit an elegant combinatorial description in terms of Maya diagrams.
\end{example}

\section{The Bochner problem: classical and exceptional polynomials}\label{sec:Bochner}

After having reviewed the notion of Darboux-Crum transformations, in this section we will introduce the concept of exceptional orthogonal polynomials, as orthogonal polynomial systems that arise from Sturm-Liouville problems with {\it exceptional degrees}, i.e. gaps in their degree sequence. But before we do so, we need to review some basic facts about Sturm-Liouville problems, and introduce Bochner's theorem, that characterizes the classical orthogonal polynomial systems of Hermite, Laguerre and Jacobi as polynomial eigenfunctions (with no missing degrees) of a Sturm-Liouville problem .

\subsection{Sturm-Liouville problems}
A Sturm-Liouville problem is a second-order boundary value problem of the form
\begin{gather}
  \label{eq:SLP}
  -(P(z) y')' + R(z) y = \lambda W(z) y,\quad y=y(z),\\
  \label{eq:SLPBC}
  \begin{aligned}
    \alpha_0 y(a) + \alpha_1 y'(a) = 0\\
    \beta_0 y(b) + \beta_1 y'(b) = 0
  \end{aligned}
\end{gather}
where $I=(a,b)$ is an interval, where $\lambda$ is a spectral
parameter, where $P(z), W(z), R(z)$ are suitably smooth real-valued
functions with $P(z), W(z)>0$ for $z\in I$.\footnote{In the case of an
  unbounded interval with $a=-\infty$ and/or $b=+\infty$, or if
  solutions $y(z)$ of \eqref{eq:SLP} have no defined value at the
  endpoints, one has to consider the asymptotics of the corresponding
  solutions and impose boundary conditions of a more general form:
\[
  \begin{aligned}
    \alpha_0(z) y(z) + \alpha_1(z) y'(z) \to 0& \quad\text{as } z\to a^{-}\\
    \beta_0(z) y(z) + \beta_1(z) y'(z) \to 0& \quad\text{as } z\to
    b^{+}
  \end{aligned}
\]
where $\alpha_0(z),\alpha_1(z), \beta_0(z), \beta_1(z)$ are continuous
functions defined on $I$.}

Dividing \eqref{eq:SLP} by $W(z)$ re-expresses the underlying
differential equation in an operator form:
\begin{equation}
  \label{eq:Tylambda}
  -T[y] = \lambda y,
\end{equation}
where
\begin{equation}
  \label{eq:Typqr}
  T[y] = p(z) y'' + q(z) y' + r(z) y,
\end{equation}
and where
\begin{equation}
  \label{eq:PWRpqr}
  \begin{aligned}
    p(z) &= \frac{P(z)}{W(z)}\qquad &
    P(z) &= \exp \int \frac{q(z)}{p(z)} dz\\
    q(z) &= \frac{P'(z)}{W(z)} &
    W(z) &= \frac{P(z)}{p(z)},\\
    r(z) &= -\frac{R(z)}{W(z)} &
    R(z)&= -r(z) W(z)
  \end{aligned}
\end{equation}

If $y_1(z), y_2(z)$ are two sufficiently smooth real-valued functions,
then integration by parts gives Lagrange's identity:
\begin{equation}
  \label{eq:Lagrange}
  \int (T[y_1] y_2 - T[y_2] y_1)(z) \,W(z) dz  = P(z) (y_1'(z) y_2(z)
  - y_2'(z) y_1(z)). 
\end{equation}
Suppose that the boundary conditions entail (i) the square
integrability of eigenfunctions with respect to $W(z)dz$ over the
interval $I$; and (ii) the vanishing of the right side of
\eqref{eq:Lagrange} at the endpoints of the interval. With some
suitable regularity assumptions on $P(z), W(z), R(z)$ one can then
show that the eigenvalues of $-T$ can be ordered so that
$\lambda_1<\lambda_2<\cdots < \lambda_n < \cdots \to \infty$.

If $y_i, y_j,\; i\neq j$ are two eigenfunctions corresponding to eigenvalues
$\lambda_i,\lambda_j$, respectively, then
\eqref{eq:Lagrange} reduces to
\begin{equation}
  \label{eq:forthog}
  (\lambda_i-\lambda_j) \int_I y_i(z) y_j(z) W(z) dz = P(z) (y_i'(z)
  y_j(z) - y_j'(z)y_i(z))\Big|^{b^-}_{a^+} = 0. 
\end{equation}
Therefore, the eigenfunctions are orthogonal with respect to the inner
product
\[ \left< f,g\right>_W = \int_I f(z) g(z) W(z) dz.\]

\begin{example}
  Let's work out the weight and boundary conditions for the Hermite
  differential equation
  \begin{equation}
    \label{eq:hde0}
     y'' - 2z y +\lambda y,\quad y=y(z).
  \end{equation}
  We apply \eqref{eq:PWRpqr} and rewrite the above in Sturm-Liouville
  form 
  \begin{equation}
    \label{eq:HSLP}
     -(W(z) y')' = \lambda W(z) y,\quad y\in \rL^2(\R,Wdz)
  \end{equation}
  where the weight has the form
  \[ W(z) = \exp \lp\int^z (-2z) \rp= e^{-z^2}\]
  In this case, the boundary conditions are that $e^{-z^2} y(z)^2$ be
  integrable near $\pm \infty$.  

  A basis of solutions to \eqref{eq:hde0} are
  \begin{align}
    \phi_0(z;\lambda)  &= \Phi\lp -\frac{\lambda}{4},\frac12,z^2\rp\\
    \phi_1(z;\lambda) &= z\,\Phi\lp \frac12 -\frac{\lambda}{4}, \frac32,z^2\rp
  \end{align}
  where
  \[ \Phi(a,c,x) = \sum_{n=0}^\infty \frac{ (a)_n}{(c)_n n!} x^n,\]
  is the confluent hypergeometric function.   This function has the
  asymptotic behaviour
  \[ \Phi(a,c,x) = \frac{\Gamma(c)}{\Gamma(a)} e^x x^{a-c}\lp 1+ O(|x|^{-1})
  \rp,\quad x\to +\infty, \]
  This implies that 
  \begin{align*}
    e^{-z^2}\phi_0(z;\lambda)^2 
    &= \frac{\pi e^{z^2}  z^{-2-\lambda}}{\Gamma(-\lambda/4)^2} \lp 1
      + O(z^{-2}) \rp,\quad z \to \pm \infty,\\ 
    e^{-z^2} \phi_1(z;\lambda)^2 
    &=  \frac{\pi   e^{z^2}  z^{-2-\lambda}}{4\Gamma(1/2-\lambda/4)^2}
      \lp 1 + O(z^{-2})\rp,\quad  z\to\pm \infty. 
  \end{align*}
  are not integrable for generic values of $\lambda$ near $z=\pm
  \infty$.
  We now introduce two other solutions of \eqref{eq:hde0},
  \begin{align}
    \psi_R(z;\lambda) 
    &= \Psi\lp -\frac{\lambda}{4},\frac12,z^2\rp,\quad z>0 \\
    \psi_L(z;\lambda) 
    &= \Psi\lp -\frac{\lambda}{4},\frac12,z^2\rp,\quad z<0    
  \end{align}
  where
  \begin{equation}
    \label{eq:psiphi}
    \Psi(a,c;x) = \frac{\Gamma(1-c)}{\Gamma(a-c+1)}\Phi(a,c;x) +
    \frac{\Gamma(c-1)}{\Gamma(a)} \Phi(a-c+1,2-c;x),\quad x>0.
  \end{equation}
  Note that $\psi_R(z)$ and $\psi_L(z)$ are different functions,
  because $\Psi$ is a branch of a multi-valued function defined by
  taking a branch cut over the negative real axis.  However,
  $\psi_L, \psi_R$ may be continued to solutions of \eqref{eq:hde0}
  over all of $\R$ by means of connection formulae \eqref{eq:psiphi1},
  below.

  We have the asymptotics
  \begin{align*}
    x^a \Psi(a,c;x) 
    &= 1+ O(x^{-1}),\quad x\to +\infty\\
    e^{-z^2} \psi_R(z;\lambda)^2 
    &=  e^{-z^2} z^\lambda\lp1 +  O(z^{-2})\rp,\quad
      z\to +\infty\\
    e^{-z^2}  \psi_L(z;\lambda)^2
    &= e^{-z^2}z^\lambda\lp1 + O(z^{-2})\rp,\quad z\to -\infty    
  \end{align*}
  Hence, $\psi_R, \psi_L$ each satisfy a one-sided boundary conditions
  at $\pm\infty$.

  From \eqref{eq:psiphi} we get the connection formulae
  \begin{equation}
    \label{eq:psiphi1}
  \begin{aligned}
    \psi_R(z;\lambda) 
    &= \frac{\sqrt{\pi}}{\Gamma(1/2-\lambda/4)} \phi_0(z;\lambda) 
     - \frac{2\sqrt{\pi}}{\Gamma(-\lambda/4)} \phi_1(z;\lambda),\\
    \psi_L(z;\lambda) 
    &= \frac{\sqrt{\pi}}{\Gamma(1/2-\lambda/4)} \phi_0(z;\lambda) 
      + \frac{2\sqrt{\pi}}{\Gamma(-\lambda/4)} \phi_1(z;\lambda).
  \end{aligned}
  \end{equation}

  Therefore, our boundary conditions amount to imposing the condition
  that $\psi_L$ be proportional to $\psi_R$.  By inspection of
  \eqref{eq:psiphi1}, this can happen in exactly two ways:
  $\psi_L = \psi_R$ and $\psi_L=-\psi_R$.  The first case occurs when
  $\Gamma(-\lambda/4) \to \infty$ that is when
  $\lambda/2 = 2n,\; n=1,2,\ldots$. The second possibility occurs when
  $\Gamma(1/2-\lambda/4) \to \infty$ which occurs when
  $\lambda/2 = 2n+1,\; n=1,2,\ldots$.  In the first case, we recover
  the even Hermite polynomials; in the second the odd Hermite
  polynomials.  This last observation can be restated as the following
  identity
  \[2^{-n} H_n(z) = \sqrt{\pi}\lp \frac{\phi_0(z;2n)}{\Gamma(1/2-n/2)}
  - \frac{2\phi_1(z;2n)}{\Gamma(-n/2)}\rp,\quad n =0,1,2,\ldots. \]
  Therefore the Hermite polynomials are precisely the solutions of
  \eqref{eq:hde0} that satisfy the boundary conditions of
  \eqref{eq:HSLP}, namely they are the only solutions of
  \eqref{eq:hde0} that are square-integrable with respect to
  $e^{-z^2}$ over all of $\R$.
\end{example}

\subsection{Classical Orthogonal Polynomials}
The notion of a Sturm-Liouville system with polynomial eigenfunctions
is the cornerstone idea in the theory of classical orthogonal
polynomials.  The reason is simple: if the eigenfunctions of a
Sturm-Liouville problem \eqref{eq:SLP} are polynomials, then they will
be orthogonal with respect to the corresponding weight $W(z)$.

The following three
types of polynomials --- bearing the names of Hermite, Laguerre, and
Jacobi --- are known as the classical orthogonal polynomials.
\begin{itemize}
\item Hermite polynomials obey the following 3-term recurrence relation:
\begin{equation}
  \label{eq:hermrr}
  2z H_n = H_{n+1} + 2n H_{n-1},\qquad     H_{-1}=0,\;H_0=1.
\end{equation}
They are orthogonal with respect to 
\[ W_{\mathrm{H}}(z) = e^{-z^2},\quad z\in(-\infty,\infty),\]
and satisfy the following differential equation
\begin{equation}
  \label{eq:hermde}
  y'' - 2z y'+2n y=0,\quad y=H_n(z),\quad n\in \N,
\end{equation}

\item Laguerre polynomials $L_n = L^{(\alpha)}_n(z)$ have one parameter
$\alpha$, and satisfy the following 3-term recurrence relation:
{\small\begin{align}
  \label{eq:lagrr}
         2z L_n 
         &=  (n+1)L_{n+1}-(2n+\alpha+1)L_n
           + (n+\alpha)  L_{n-1},\qquad     L_{-1}=0, L_0=1.
\end{align}}
\noindent
For $\alpha>-1$, Laguerre polynomials are orthogonal with respect to
\[ W_{\mathrm{L}} = e^{-z}z^\alpha,\quad z\in (0,\infty).\]
They satisfy the following differential equation
\begin{equation}
  \label{eq:lagde}
  z y'' +(\alpha+1-z) y'+ ny=0
  ,\quad y=L^{(\alpha)}_n(z),\quad n\in \N,
\end{equation}

\item Jacobi polynomials $P_n = P^{(\alpha,\beta)}_n(z)$ have two
parameters, $\alpha,\beta$ and are defined by: {\small\begin{align} z
    P_n=& \frac{2 (n+1) (n+\alpha +\beta +1)}{(2 n+\alpha +\beta +1)
      (2 n+\alpha +\beta +2)}P_{n+1}\\ \nonumber
    &+\frac{\left(\beta ^2-\alpha
      ^2\right)}{(2 n+\alpha +\beta ) (2 n+\alpha +\beta +2)}P_n \\
                                                        \nonumber
        &+\frac{2 (n+\alpha ) (n+\beta )}{(2 n+\alpha +\beta ) (2
          n+\alpha +\beta +1)}P_{n-1},  && P_{-1}=0,\; P_0=1 
\end{align}}
These polynomials obey the differential equation
\begin{equation}
  \label{eq:jacde}
  (1-z^2) P_n''+(\beta-\alpha-z (\alpha+\beta+2)) P_n'+n
  (\alpha+\beta +n+1) P_n=0
\end{equation}
For $\alpha,\beta>-1$ they are orthogonal with respect to
\[ W_{\mathrm{H}} = (1-z)^{\alpha} (1+z)^\beta ,\quad z\in(-1,1).\]

\end{itemize}

\begin{shaded}
\begin{exercise}
  Rewrite the above differential equations in Sturm-Liouville form.
  In each case,  work out the boundary conditions that pick out the
  polynomial solutions.
\end{exercise}
\end{shaded}

The class of Sturm-Liouville problems with polynomial eigenfunctions
 was studied and classified by Solomon Bochner in the following
fundamental result.  Bochner's Theorem was subsequently refined by
Lesky to show that the three classical families of Hermite, Laguerre,
and Jacobi give a full classification of such Sturm-Liouville
problem.

\begin{thm}[Bochner]
  \label{thm:bochner}
  Suppose that an operator 
  \begin{equation}
    \label{eq:Tpqr}
    T[y]=p(z) y'' + q(z) y' + r(z) y
  \end{equation}
  admits eigenpolynomials of every degree; that is, there exist
  polynomials $y_k(z)$ with $\deg y_k = k$ and constants $\lambda_k$
  such that
  \begin{equation}
    \label{eq:Tyklambdak}
     -T[y_k] = \lambda_k y_k,\quad k=0,1,2,\ldots.
  \end{equation}
  Then, necessarily $p,q,r$ are polynomials with 
  \[ \deg p \leq 2,\quad \deg q\leq 1,\quad \deg r= 0.\]
  Moreover, if these polynomials are the orthogonal eigenfunctions of
  a Sturm-Liouville system, then up to an affine transformation of the
  independent variable $z$, they are the classical polynomials of
  Hermite, Laguerre, and Jacobi.
\end{thm}
\begin{proof}
  Applying \eqref{eq:Tyklambdak} to $k=0,1,2$, we obtain
  \begin{align*}
    -\lambda_0 y_0&= r\\
    -\lambda_1 y_1&= q y_1' + r y_1\\
    -\lambda_2 y_2&= p y_2'' + q y_2' + r y_2 .
  \end{align*}
  By inspection, $r$ is a constant, while $q,p$ are polynomials with
  $\deg q\leq 1$ and $\deg p\leq 2$.

  Up to an affine transformation $z\mapsto s z + t$, the leading
  coefficient $p(z)$ can assume one of the following normal forms:
  \[ 1, z, z^2, 1-z^2, 1+z^2.\]
  Write $q(z) = az + b$.  Applying \eqref{eq:PWRpqr}, the
  corresponding weights have the form
  \begin{align*}
    \text{(i)}\qquad  & W(z)  = e^{\frac{b^2}{2a}} e^{\frac{a}{2} (z+b/a)^2}\\
    \text{(ii)}\qquad & W(z) =  e^{az} z^{b-1}\\
    \text{(iii)}\qquad & W(z) =  e^{-\frac{b}{z}}z^{a-2},\\
    \text{(iv)}\qquad & W(z) = (1-z)^{-(a+b)/2-1} (1+z)^{(b-a)/2-1}\\
    \text{(v)}\qquad & W(z) = e^{b \arctan(z)} (1+z^2)^{a/2-1}.
  \end{align*}
  \begin{itemize}
  \item   For normal form (i), the case $a=0$ is excluded.  If not, the
  resulting operator would be strictly degree lowering, which would
  preclude the existence eigenpolynomials of degrees $\geq 2$. Since
  $p(z)=1$ is invariant with respect to scaling and translation, no
  generality is lost by setting $b=0$, $a=\pm 2$.  The case of $a=-2$
  corresponds to the classical Hermite polynomials.  The case $a=2$
  can be excluded because there is no choice of boundary conditions
  that result in the vanishing of the right side of \eqref{eq:forthog}.

  \item For the normal form (ii), note that $p(z)=z$ is preserved by scaling transformations.  Hence, without loss of
  generality we can take $a=-1$.  This case corresponds to the classical Laguerre
  polynomials.

  \item Normal form (iii) is a bit tricky.  The case $b= 0$ can be ruled out
  because of the absence of suitable boundary conditions.  The
  analysis of $b<0$ and $b>0$ is the same, so suppose that $b>0$.
  Here the only possible boundary conditions are at the endpoints of
  the interval $(0,\infty)$. If $a<0$ then a finite number of
  polynomials can be made to be square integrable with respect to the
  weight in question.  These constitute the so-called Bessel
  orthogonal polynomials, which however fall outside the range of our
  definition --- we require that \emph{all} $y_k$ are
  square-integrable with respect to $W(z)$.

  \item Normal form (iv) corresponds to the Jacobi orthogonal polynomials.

  \item Normal form (v) corresponds to the so-called twisted Jacobi (also called Romanovsky)
  polynomials.  If $a<0$ then a finite number of
  initial degrees   are square-integrable with respect to the
  indicated weight over the interval $(-\infty,\infty)$.  As above,
  this violates our requirement that \emph{all} the $y_k$ be
  square-integrable with respect to $W(z) dz$.
  \end{itemize}
\end{proof}

\subsection{Exceptional Polynomials and Operators}

We now modify the assumption of Bochner's Theorem \ref{thm:bochner} to
arrive at the following.  
\begin{defn}
  We say that $T[y]=p(z) y'' + q(z) y' + r(z) y$ is an
  \emph{exceptional operator} if it admits polynomial eigenfunctions
  for a cofinite number of degrees; that is, there exist polynomials
  $y_k(z),\; k\notin \N\setminus \{ d_1,\ldots, d_m \}$ with
  $\deg y_k = k$ and with $ d_1,\ldots, d_m\in \N$ a finite number of
  exceptional degrees, and constants $\lambda_k$ such that
  \[ -T[y_k] = \lambda_k y_k,\quad k\in \N\setminus\{ d_1,\ldots, d_m
  \}.\]
  Moreover, if it is possible to impose boundary conditions so that the
  polynomials $y_k$ become eigenfunctions of the corresponding
  Sturm-Liouville problem, then we call the
  $\{y_k\}_{k\notin \{d_1,\ldots, d_m \}}$ \emph{exceptional
    orthogonal polynomials}.
\end{defn}

The relaxed assumption that permits for a finite number of missing
degrees allows to escape the constraints of Bochner's theorem and
characterizes a large and interesting new class of operators and
polynomials. 

\begin{example} We next show an example of codimension 2 exceptional Hermite polynomials.
Recall the classical Hermite polynomials defined by \eqref{eq:hermrr}.
Introduce a family of exceptional Hermite polynomials defined by
\begin{equation}
  \label{eq:hHdef}
  \hH_n = \frac{\Wr[H_1,H_2,H_n]}{8(n-1)(n-2)} = H_n+ 4 n H_{n-2}
  + 4 n(n-3) H_{n-4},\quad n\neq 1,2
\end{equation}
where the $H_i(z)$ are classical Hermite polynomials and where $\Wr$
denotes the usual Wronskian determinant:
\[ \Wr[H_1,H_2,H_n] = 
\begin{vmatrix}
  H_1 & H_1' & H_1'' \\
  H_2 & H_2' & H_2''\\
  H_n & H_n' & H_n''
\end{vmatrix}.
\]

\end{example}
\begin{exercise}
  Using the following identity for the classical Hermite polynomials:
  \[ H_n' = 2n H_{n-1}\]
and the 3-term recurrence relation \eqref{eq:hermrr} reduce the
Wronskian expression
\[ \hH_n = \frac{\Wr[H_1,H_2,H_n]}{8(n-1)(n-2)},\quad n\neq 1,2\]
to the right-hand side expression shown in \eqref{eq:hHdef}.
\end{exercise}
Observe that $\deg \hH_n = n$.  We call the resulting sequence of
polynomials \emph{exceptional} because the degree sequence
$\deg \hH_n$ is missing two the degrees -- the exceptional degrees
$n=1$ and $n=2$. We call the $\hH_n(z)$ \emph{exceptional Hermite}
polynomials because they furnish polynomial solutions of the following
modified version of the Hermite differential equation:
\begin{equation}
  \label{eq:EHDE}
   y'' - \left(2z +\frac{8z}{1+2z^2}\right) y' +2n   y=0,\quad y=
   \hH_n(z),\; n\neq 1,2.  
\end{equation}

At first glance, the exceptional modification of Hermite's differential equation
\eqref{eq:EHDE} has a rather peculiar form; indeed it is slightly
paradoxical that a differential equation with rational coefficients
admits polynomial solutions.  However, some of the underlying
structure of the equation comes to light once we ``clear
denominators'' and re-express \eqref{eq:EHDE} using the following, bilinear
formulation:
\begin{equation}
  \label{eq:EHDEbilin}
  \left(\eta y'' - 2\eta' y' + \eta'' y\right) -2z \left(
        \eta y' - \eta' y\right) + 2(n-2)\, \eta y = 0
\end{equation}
where 
\[ \eta= \Wr[H_1,H_2]=4+8z^2 \]
Now the equation is bilinear in $\eta$, which is fixed and $y=y(z)$
the dependent variable, and nearly symmetric with respect to the two
variables.  

We can also rewrite expression \eqref{eq:EHDE} using Sturm-Liouville form, as
\begin{equation}
  \label{eq:EHDESL}
  \left( \hW y'\right)' = \lambda \hW y,
\end{equation}
where
\[ \hW(z) = \frac{e^{-z^2}}{\eta(z)^2},\quad \lambda =-2n. \]
As before, the Sturm-Liouville form implies the orthogonality of the
eigenpolynomials:
\[ \int_{-\infty}^\infty \hH_m(z) \hH_n(z) \hW(z) dz= 0,\quad m\neq
n,\; m,n\neq 1,2\]
It is also possible to show that the exceptional polynomials satisfy
recurrence relations.  However, now there are multiple relations of
higher order:
{ \small
\begin{align}
  \label{eq:rr3}
  4 z(3+2z^2) \hH_n &= \hH_{n+3} + 6n\, \hH_{n+1} + 12 n(n-3)\,
  \hH_{n-1} + 8 n(n-4)(n-5)\, \hH_{n-3},\\
  \label{eq:rr4}
  16z^2(1+z^2) \hH_n &= \hH_{n+4} + 8 n \hH_{n+2} + 4 (6n^2-14n+1)
  \hH_n + \\ \nonumber
  &\qquad + 32 n(n-3)(n-4) \hH_{n-2} + 16 n(n-3)(n-5)(n-6) \hH_{n-4}
\end{align}
} 

\begin{table}
  \centering
  \begin{tabular}{c|c|c}
    $n$ & RHS Degrees in relation \eqref{eq:rr3}  
    & RHS Degrees in
      relation \eqref{eq:rr4} \\ \hline
    \strut 0 & 3 & 4,0\\
    3 & 6,4,0 &  7,5,3\\
    4 & 7,5,3 &  8,6,4,0\\
    5 & 8,6,4 &  9,7,5,3\\
    6 & 9,7,5,3 & 10,8,6,4\\
    7 & 10,8,6,4 & 11,9,7,5,3\\
    $\vdots$ & $\vdots$ & $\vdots$ \\
    $n\geq 7$ & $n+3, n+1, n-1, n-3$ & $n+4,n+2,n,n-2,n-4$
  \end{tabular} 
  \caption{Degrees in the exceptional recurrence relations}
  \label{tab:x2deg}
\end{table}
\noindent
Table \ref{tab:x2deg} lists the degrees of the exceptional polynomials
involved in the above recurrence at the values of $n=0,3,4,\ldots$.
By inspection, $\hH_0$ determines
$\hH_3, \hH_4, \hH_6, \ldots, \hH_{2k},$ $k\geq 2.$ Relation
\eqref{eq:rr3} with $n=5$ then determines $\hH_5$.  After that the
$\hH_{2k+1},\; k\geq 3$ are established.  Observe that
$\hH_n(z),\; n\geq 7$ are determined by both relations \eqref{eq:rr3}
and \eqref{eq:rr4}.  Remarkably, the relations are coherent, in the
sense that both relations give the same value of
$\hH_n(z),\; n\geq 7$.  This may be explained by the fact that the
finite-order difference operators that describe the RHS of
\eqref{eq:rr3} and \eqref{eq:rr4} commute with one another.
\begin{shaded}
\begin{exercise}
  Verify the recurrence relations  \eqref{eq:rr3} and \eqref{eq:rr4} using a computer algebra
  system and show that the  finite-difference operators that define their
  right-hand sides commute.
\end{exercise}
\end{shaded}

Finally, many of the properties of exceptional polynomials are
explained by the fact that there is a hidden relation between them and
their classical counterparts.  Let us define second order operators
\begin{align*}
  T[y] &= y'' - 2z y',\\
  \hT[y] &= y'' - \left(2z +\frac{8z}{1+2z^2}\right) y' 
\end{align*}
and re-express the classical and exceptional Hermite differential
equations in operator form, respectively, as
\[ -T[H_n] =  2n H_n,\; n\in \N \quad   -\hT[\hH_n] = 2n \hH_n,\;
n\neq 1,2.\]
Let us also introduce the second order operator
\[ A[y] = \Wr[H_1, H_2, y] =  4 (1+2z^2) y''-16z y'+16 y.\]

\begin{shaded}
\begin{exercise}
  Verify that the three differential  operators $T,\hT$ and $A$  satisfy the following (second-order)
intertwining relation:
\begin{equation}
  \label{eq:EHDEIR}
 \hT A = A T.
\end{equation}
\end{exercise}
\end{shaded}
Note that in the intertwining relations \eqref{eq:intertwining}, operator $A$ is first order, corresponding to a single-step Darboux transformation. In this case, $A$ is a second-order differential operator that comes from a 2-tep Darboux-Crum transformation with seed functions $H_1$ and $H_2$. 
In general, up to a normalization constant, the exceptional polynomials
are given by applying the intertwiner $A$ to the classical
polynomials:
\[ \hH_n \propto A[H_n].\]
If we take the intertwining relation as proven, we obtain
that
\[ \hT[A[H_n]] = (\hT A)[H_n] = (A T)[H_n] = -2n A[H_n] .\]
Thus, the intertwining relation ``explains'' why the $\hH_n$ are
eigenpolynomials of the exceptional operator $\hT$. This is essentially the same argument as the one used in Exercise 1, albeit with a higher-order intertwiner $A$.

\section{Symmetric Painlev\'e equations and Darboux dressing chains}\label{sec:symP}

And now for something completely different \cite{montypython}, or maybe not ?
The set of six nonlinear second order \p\ equations ${\rm P_I,\dots,P_{VI}}$ have attracted considerable interest in the past 100 years \cite{clarkson2003painleve,gromak2008painleve}.
They have the defining property that their solutions have no movable branch points.
The \p\ equations, whose solutions are called \p\ transcendents, are now considered to be the nonlinear analogues of special functions, cf.~\cite{clarkson2003painleve}. 
These functions, in general, are transcendental in the sense that they cannot be expressed in terms of previously known functions. However, the \p\ equations, except $\Pone$, also possess special families of solutions that can be expressed via rational functions, algebraic functions or the classical special functions, such as Airy, Bessel, parabolic cylinder, Whittaker or hypergeometric functions, for special values of the parameters.

However, rather than studying the Painlev\'e second order scalar equations, we will follow Noumi and Yamada since it will prove to be more useful to rewrite these equations as a system of first order equations, which will allow us  not only to understand the symmetry properties better, but also to generalize these system to higher order equations with the same desired properties.

\begin{defn}
We define the $A_2$-Painlev\'e system as the following system of three coupled nonlinear ODEs
\begin{eqnarray}\label{eq:P4system}
f_0' + f_0(f_1-f_2) &=& \alpha_0, \nonumber\\
f_1' + f_1(f_2-f_0) &=& \alpha_1,\\
f_2' + f_2(f_0-f_1) &=& \alpha_2, \nonumber 
\end{eqnarray}
subject to the condition
\begin{equation}\label{eq:P4normalization}
(f_0+f_1+f_2)'=\alpha_0+\alpha_1+\alpha_2=1.
\end{equation}
where $\alpha_0, \alpha_1, \alpha_2\in\mathbb C$ are complex parameters and $f_i=f_i(z)$ are complex functions.
\end{defn}
If the parameters take on arbitrary values, the general solution of this equation is transcendental. We are interested in this lecture to find solutions to \eqref{eq:P4system} where the functions $f_i=f_i(z)$ are rational functions of $z$. A solution of \eqref{eq:P4system} will be a tuple of the form $(f_0,f_1,f_2 |\alpha_0,\alpha_1,\alpha_2)$.

The reason why this system is relevant is that by eliminating two of the functions, we can reduce system \eqref{eq:P4system} to a single second order nonlinear ODE, that we will call $\Pfour$ because it is the fourth equation in the list of six Painlev\'e equations, namely:

\begin{equation}\label{eq:P4scalar}
y''=\frac{1}{2y}(y')^2+\frac{3}{2} y^3 + 4t y^2 +2(t^2-a)y+\frac{b}{y}
\end{equation}
\begin{shaded}
\begin{exercise}
Show that if the tuple $(f_0,f_1,f_2 |\alpha_0,\alpha_1,\alpha_2)$ is a solution to \eqref{eq:P4system}, then $y=y(t)$ is a solution to \eqref{eq:P4scalar}, where:
\begin{equation}\label{eq:reparam}
f_0=-c y,\quad z=-\frac{t}{c},\quad c=\sqrt{\frac{-1}{2}},\quad a=2(\alpha_1-\alpha_2),\quad b=-2\alpha_0^2
\end{equation}
\end{exercise}
\end{shaded}
We first take the derivative of the first equation in  \eqref{eq:P4system}:
\begin{equation}
f_0''+f_0'(f_1-f_2)+f_0(f_1'-f_2')=0
\end{equation}
Next subtract the third from the second equation to obtain
\[f_1'-f_2'=\alpha_1-\alpha_2-2f_1f_2+f_0(f_1+f_2)  \]
and insert it into the previous equation, to get
\begin{equation}\label{eq:fopp1}
f_0''+f_0'(f_1-f_2)+(\alpha_1-\alpha_2) f_0-2f_1f_2f_0+(f_1+f_2)f_0^2=0.
\end{equation}
From the first equation in \eqref{eq:P4system} and the normalization $f_0+f_1+f_2=z$, we have
\begin{eqnarray}\label{eq:2int1}
f_1-f_2&=&\frac{\alpha_0-f_0'}{f_0}\\
f_1+f_2&=&z-f_0 \label{eq:2int2}
\end{eqnarray}
Now bearing in mind that $4f_1f_2= (f_1+f_2)^2-(f_1-f_2)^2$ we have also
\begin{equation} \label{eq:4f1f2}
4f_1f_2=(z-f_0)^2-\left(\frac{\alpha_0-f_0'}{f_0} \right)^2
\end{equation}
Inserting \eqref{eq:2int1} ,\eqref{eq:2int2} and \eqref{eq:4f1f2} into \eqref{eq:fopp1}, and after some cancellations and grouping terms we arrive at
\begin{equation}
f_0''=\frac{f_0'^2}{2f_0}+\frac{3}{2}f_0^3-2zf_0^2+\left(\frac{z^2}{2}+\alpha_2-\alpha_1 \right)f_0-\frac{\alpha_0^2}{2f_0}
\end{equation}
which after the rescaling of variable, function and parameters shown in \eqref{eq:reparam} leads finally to \eqref{eq:P4scalar}.

Now that we know the equivalence between solutions of \eqref{eq:P4system}, that we will call ${\rm s}\Pfour$, the symmetric form of $\Pfour$, it will be easier to work with the system than with the equation. In particular, Noumi and Yamada showed \cite{noumi1999symmetries} that system \eqref{eq:P4system} in invariant under a symmetry group, which acts by B\"acklund transformations on a tuple of functions and parameters. This symmetry group is
the affine Weyl group $A_2^{(1)}$, generated by the operators
$\{ \mathbf{\pi},\textbf{s}_0, \textbf{s}_1, \textbf{s}_2\}$ whose
action on the tuple $(f_0,f_1,f_2 |\alpha_0,\alpha_1,\alpha_2)$ is
given by:
\begin{eqnarray}\label{eq:BT}
&&{\bf s}_k (f_j)=f_j-\frac{\alpha_k\delta_{k+1,j}}{f_k}+\frac{\alpha_k\delta_{k-1,j}}{f_k}, \nonumber\\
&&{\bf s}_k(\alpha_j)=\alpha_j-2\alpha_j\delta_{k,j}+\alpha_k(\delta_{k+1,j}+\delta_{k-1,j}),\\
&&{\bf \pi}(f_j)=f_{j+1},\qquad  {\bf \pi}(\alpha_j)=\alpha_{j+1} \nonumber
\end{eqnarray}
where $\delta_{k,j}$ is the Kronecker delta and $j,k=0,1,2 \mod(3)$.

The technique to generate rational solutions is to first identify a
number of very simple rational \textit{seed solutions}, and then
successively apply the B\"acklund transformations \eqref{eq:BT} to
generate families of rational solutions.

\begin{shaded}
\begin{exercise}
Check that the tuple $(z,0,0 |1,0,0)$ satisfies \eqref{eq:P4system}. This is one possible \textit{seed solution}. Now use the B\"acklund transformations ${\bf s}_0$ and ${\bf s}_1{\bf s}_0$  to generate two solution tuples, and check explicitly that the obtained solutions solves \eqref{eq:P4system}
\end{exercise}
\end{shaded}
It is obvious that $(z,0,0 |1,0,0)$ satisfies \eqref{eq:P4system}. From \eqref{eq:BT}, the action of ${\bf s}_0$ on the generic tuple $(f_0,f_1,f_2 |\alpha_0,\alpha_1,\alpha_2)$ is given by
\begin{align}
{\bf s}_0(f_0)&=f_0,& &{\bf s}_0(\alpha_0)=-\alpha_0\\
{\bf s}_0(f_1)&=f_1-\frac{\alpha_0}{f_0},& &{\bf s}_0(\alpha_1)=\alpha_1+\alpha_0\\
{\bf s}_0(f_2)&=f_2+\frac{\alpha_0}{f_0},& &{\bf s}_0(\alpha_0)=\alpha_2+\alpha_0
\end{align}
So we have then that ${\bf s}_0(z,0,0 |1,0,0)=\left(z,\frac{-1}{z},\frac{1}{z} |-1,1,1\right)$, and we can readily verify that this tuple satisfies \eqref{eq:P4system}.
In a similar manner, we see that
\[{\bf s}_1 {\bf s}_0(z,0,0 |1,0,0)={\bf s}_1\left(z,\frac{-1}{z},\frac{1}{z} \Big|-1,1,1\right)=\left(0,-\frac{1}{z},z+\frac{1}{z}\,\Big|\,0,-1,2\right) \]
which is also seen to satisfy \eqref{eq:P4system}.

In this way we can iteratively apply B\"acklund transformations on a small set of seed solutions and generate many rational solutions to \eqref{eq:P4system}.
This is a beautiful approach, pioneered by the japanese school, and the transformations \eqref{eq:BT} have a nice geometric interpretation in terms of reflection groups acting on the space of parameters $(\alpha_0,\alpha_1,\alpha_2)$. Note however that the solutions obtained by dressing a given seed solution are hard to write in closed form, and in general the whole procedure is more an algorithm to generate solutions than an explicit enumeration of them. If we ask ourselves how many poles the rational solution ${\bf s}_1^6\, {\bf s}_0^3(z,0,0 |1,0,0)$ has, this might be a difficult question to answer with this representation.

For this reason, we will not pursue this approach henceforth in these notes, and we refer the interested reader to Noumi's book \cite{noumi2004painleve} to learn the geometric theory of Painlev\'e equations, and their connections with other topics in integrable systems ($\tau$-functions, Hirota bilinear equations, Jacobi-Trudi formulas, reductions from KP equation, etc.).

We will concentrate in these lectures on alternative representations of the rational solutions, most notably the determinantal representations \cite{kajiwara1996determinant, kajiwara1998determinant}.

Once we are aware of the symmetry structure of \eqref{eq:P4system}, the system admits a natural generalization to any number of equations, known as
the $A_{N}^{(1)}$-Painlev\'e or the Noumi-Yamada system. The even case ($N=2n$)  is considerably simpler (for reasons that will be
explained later), and it is the one we will focus on this notes.

\begin{defn}
We define the $A_{2n}^{(1)}$-Painlev\'e system (or Noumi-Yamada system) as the following system of $2n+1$ coupled nonlinear ODEs
\end{defn}
\begin{equation}\label{eq:Ansystem}
  f_i'+f_i \left( \sum_{j=1}^n f_{i+2j-1} - \sum_{j=1}^n f_{i+2j}
  \right)=\alpha_i,\qquad i=0,\dots,2n \mod (2n+1)
\end{equation}
subject to the normalization condition
\begin{equation}
  \label{eq:alpha1}
(f_0+\dots+f_{2n})'=\alpha_0+\cdots + \alpha_{2n} =1.
\end{equation}

The symmetry group of this higher order system is the affine Weyl group $A_{2n}^{(1)}$, acting by B\"acklund transformations as in \eqref{eq:BT}. The system has the Painlev\'e property, and thus can be considered a proper higher order generalization of $\rm{s}\Pfour$ \eqref{eq:P4system}, which corresponds to $n=1$.

The goal of this lecture is to develop a systematic procedure to describe rational solutions to system \eqref{eq:BT}, providing an explicit representation of the solutions in terms of Wronskian determinants whose entries are Hermite polynomials. This is an alternative approach to the dressing of seed solutions by B\"acklund transformations described above.

\subsection{Darboux dressing chains}\label{sec:dressing}

The theory of dressing chains, or sequences of Schr\"{o}dinger operators
connected by Darboux transformations was developed by Adler
\cite{adler1994nonlinear}, and Veselov and Shabat
\cite{veselov1993dressing}. The connection between dressing chains and
\p\ equations was already shown in \cite{adler1994nonlinear}
and it has been exploited by some authors 
\cite{takasaki2003spectral,tsuda2005universal,bermudez2012complexb,marquette2013one,marquette2013two,marquette2016,
sen2005darboux,WilloxHietarinta,MateoNegro}. 
This section follows mostly the early works of Adler, Veselov and Shabat.

Consider the following sequence of Schr\"{o}dinger operators
\begin{equation}\label{eq:Lseq}
  L_i = -D_z^2 + U_i ,\qquad D_z= \frac{d}{dz},\quad U_i=U_i(z),\quad
  i\in \Z
\end{equation}
where each operator is related to the next by a Darboux transformation, i.e. by the following factorization
\begin{equation}
  \label{eq:Dxform}
  \begin{aligned}
    L_i &= (D_z + w_i)(-D_z + w_i)+\lambda_i, \quad w_i = w_i(z),\\
    L_{i+1} &= (-D_z + w_i)(D_z + w_i)+\lambda_i.
  \end{aligned}
\end{equation}
It follows that the functions $w_i$ satisfy the Riccati equations
\begin{equation}\label{eq:Riccati}
 w_i' + w_i^2  = U_i - \lambda_i,\quad -w_i' +w_i^2 = U_{i+1}- \lambda_i.
 \end{equation}
Equivalently, $w_i$ are the log-derivatives of $\psi_i$, the seed function of the Darboux transformation that maps $\L_i$ to $\L_{i+1}$
\begin{equation}
  \label{eq:Lipsii}
  L_i\psi_i = \lambda_i\psi_i,\qquad\text{where } w_i = \frac{\psi_i'}{\psi_i}.
\end{equation}
Using \eqref{eq:Lseq} and \eqref{eq:Dxform}, the potentials of the dressing chain are related by
\begin{align}
  U_{i+1} &= U_i - 2 w'_i, \label{eq:Uplus1}\\
  U_{i+n} &=U_i - 2 \left( w'_i+ \cdots + w'_{i+n-1}\right),\quad
  n\geq 2 \label{eq:Uplusn}
\end{align}
If we eliminate the potentials in \eqref{eq:Riccati} and set
\begin{equation}
  \label{eq:alphaidef}
  a_i =  \lambda_{i} - \lambda_{i+1}
\end{equation}
the following chain of coupled equations is obtained
\[
(w_i + w_{i+1})' + w_{i+1}^2 - w_i^2 = a_i ,\quad i\in \Z
\]
Before continuing, note that this infinite chain of equations has the
evident reversal symmetry
\begin{equation}
  \label{eq:reversal1}
  w_i \mapsto -w_{-i},\qquad a_i \mapsto -a_{-i}.
\end{equation}

This infinite chain of equations closes and becomes a finite
dimensional system of ODEs if a cyclic condition is imposed on the
potentials of the chain
\begin{equation}\label{eq:shift}
  U_{i+p} = U_i+\Delta,\quad i\in \Z
 \end{equation}
for some $p\in \N$ and $\Delta \in \C$.  If this holds, then
necessarily $w_{i+p}=w_i$, $a_{i+p}=a_i$, and
\begin{equation}
  \label{eq:Deltasumalpha}
 \Delta= -(a_0 + \cdots + a_{p-1}). 
\end{equation}

\begin{defn}\label{def:wchain}
  A $p$-cyclic Darboux dressing chain (or factorization chain) with
  shift $\Delta$ is a sequence of $p$ functions $w_0,\ldots, w_{p-1}$
  and complex numbers $a_0,\ldots, a_{p-1}$ that satisfy
  the following coupled system of $p$ Riccati-like ODEs
\begin{equation}
  \label{eq:wfchain}
  (w_i + w_{i+1})' + w_{i+1}^2 - w_i^2 = a_i ,\qquad
  i=0,1,\ldots, p-1 \mod(p)  
\end{equation}
subject to the condition \eqref{eq:Deltasumalpha}.
\end{defn}
Note that
transformation 
\begin{equation}
  \label{eq:reversal2}
  w_i \mapsto -w_{-i},\quad a_i \mapsto -a_{-i},\quad
  \Delta\mapsto -\Delta
\end{equation}
projects the reversal symmetry to the finite-dimensional system
\eqref{eq:wfchain}.  Moreover, for $j=0,1\ldots, p-1$ we also have the
cyclic symmetry
\[ w_i \mapsto w_{i + j},\quad a_i \mapsto a_{i+j},\quad
\Delta \mapsto \Delta \qquad i=0,\ldots p-1 \mod(p) \]
In the classification of solutions to \eqref{eq:wfchain} it will be
convenient to regard two solutions related by a reversal symmetry or
by a cyclic permutation as being equivalent.

Adding the $p$ equations \eqref{eq:wfchain} we immediately obtain a
first integral of the system
\[ \sum_{j=0}^{p-1} w_j= \tfrac12z\sum_{j=0}^{p-1} a_j=
-\tfrac12{\Delta}z.\]

The equivalence between the $A_{2n}$-\p\ system \eqref{eq:Ansystem} and the cyclic dressing chain \eqref{eq:wfchain} is given by the following proposition.

\begin{prop}\label{prop:wtof}
If the tuple of functions and complex numbers $(w_0,\dots,w_{2n}|a_0,\dots,a_{2n})$  satisfies a $(2n+1)$-cyclic Darboux dressing chain with shift $\Delta$ as per Definition~\ref{def:wchain}, then the tuple $\left(f_0,\dots,f_{2n}\,\big|\,\a_0,\dots, \a_{2n}\right)$
with
\begin{eqnarray}
f_i(z)&=& c \,(w_i + w_{i+1})\left(cz\right),\qquad i=0,\dots,2n\mod(2n+1),\\
 \a_i&=&c^2 a_i,\\
 c^2&=&-\frac{1}{\Delta}
\end{eqnarray}
 solves the $A_{2n}$-\p\ system \eqref{eq:Ansystem} with normalization \eqref{eq:alpha1}.
\end{prop}
\begin{proof}
The linear transformation
\begin{equation}\label{eq:wtof}
f_i=w_i+w_{i+1}, \qquad i=0,\dots,2n\mod(2n+1)
\end{equation}
is invertible (only in the odd case $p=2n+1$), the inverse transformation being
\begin{equation}\label{eq:ftow}
  w_i= \tfrac{1}{2} \sum_{j=0}^{2n} (-1)^j f_{i+j}, \qquad
  i=0,\dots,2n\mod(2n+1) 
\end{equation}
They imply the relations
\begin{equation}\label{eq:wdif}
  w_{i+1}-w_i = \sum_{j=0}^{2n-1} (-1)^j f_{i+j+1}, \qquad
  i=0,\dots,2n\mod(2n+1). 
\end{equation}
Inserting \eqref{eq:wtof} and \eqref{eq:wdif} into the equations of the cyclic dressing chain \eqref{eq:wfchain} leads to the $A_{2n}$-\p\ system \eqref{eq:Ansystem}. For any constant $c\in\mathbb C$, the scaling transformation
\[f_i\mapsto c f_i,\quad z\mapsto cz,\quad \a_i\mapsto c^2 \a_i \]
preserves the form of the equations \eqref{eq:Ansystem}. The choice $c^2=-\frac{1}{\Delta}$ ensures that the normalization \eqref{eq:alpha1} always holds, for dressing chains with different shifts $\Delta$.
\end{proof}

\begin{rem}
$(2n)$-cyclic dressing chains and $A_{2n-1}$-\p\ systems are also related, but the mapping is given by a rational rather than a linear function. A full treatment of this even cyclic case (which includes  $\Pfive$ and its higher order hierarchy) is considerably harder and shall be treated elsewhere.
\end{rem}

The problem now becomes that of finding and classifying cyclic dressing chains, i.e. Schr\"{o}dinger operators and sequences of Darboux transformations that reproduce the initial potential up to an additive shift $\Delta$ after a fixed given number of transformations.

The theory of exceptional polynomials is intimately related with families of Schr\"{o}dinger operators connected by Darboux transformations \cite{gomez2013conjecture,garcia2016bochner}. Constructing cyclic dressing chains on this class of potentials becomes a feasible task, and knowledge of the effect of rational Darboux transformations on the potentials suggests that the only family of potentials to be considered in the case of odd cyclic dressing chains are the rational extensions of the harmonic oscillator \cite{gomez2013rational}, which are exactly solvable potentials whose eigenfunctions are expressible in terms of exceptional Hermite polynomials.

Each potential in this class can be indexed by a finite set of integers (specifying the sequence of Darboux transformations applied on the harmonic oscillator that lead to the potential), or equivalently by a Maya diagram, which becomes  very useful representation to capture a notion of equivalence and relations of the type \eqref{eq:shift}.

As mentioned before, the fact that all rational odd cyclic dressing chains (and equivalently rational solutions to the $A_{2n}$-\p\ system) must \textit{necessarily} belong to this class remains an open question. We conjecture that this is indeed the case, and no  rational solutions other than the ones described in the following sections exist.

\section{Rational extensions of the Harmonic oscillator}\label{sec:Maya}

\subsection{Maya diagrams}

In this Section we construct odd cyclic dressing chains on potentials belonging to the class of rational extensions of the harmonic oscillator. Every such potential is represented by a Maya diagram, a rational Darboux transformation acting on this class will be a flip operation on a Maya diagram and cyclic Darboux chains correspond to cyclic Maya diagrams. With this representation, the main problem of constructing rational cyclic Darboux chains becomes purely algebraic and combinatorial.

Following Noumi \cite{noumi2004painleve}, we define a Maya diagram in the following manner.
\begin{defn}
  A Maya diagram is a set of integers $M\subset\Z$ that contains a
  finite number of positive integers, and excludes a finite number of
  negative integers.  We will use $\cM$ to denote the set of all Maya
  diagrams.
\end{defn}

\begin{defn}\label{def:index}
Let $m_1>m_2>\cdots$ be the elements of a Maya diagram $M$ arranged in decreasing order. By assumption, there exists a unique integer $s_M\in \Z$ such that $m_i = -i+s_M$ for all
$i$ sufficiently large. We define $s_M$ to be the index of $M$.
\end{defn}

We visualize a Maya diagram as a horizontally extended sequence of
$\boxdot$ and $\emptybox$ symbols with the filled symbol $\boxdot$ in
position $i$ indicating membership $i\in M$. The defining assumption
now manifests as the condition that a Maya diagram begins with an
infinite filled $\boxdot$ segment and terminates with an infinite
empty $\emptybox$ segment.

\begin{defn}

Let $M$ be a Maya diagram, and 
\[ M_-= \{ -m-1 \colon m\notin M, m<0\},\qquad M_+ = \{ m\colon m\in
M\,, m\geq 0 \}. \]
Let $s_1>s_2>\cdots > s_p$ and $t_1> t_2>\dots> t_q$ be the
elements of $M_-$ and $M_+$ arranged in descending order. 
 
We define the \textit{Frobenius symbol} of $M$ to be the double
list $(s_1,\ldots, s_p \mid t_q,\ldots, t_1)$.
\end{defn}
It is not hard to show that $s_M=q-p$ is the index of $M$.  The classical Frobenius symbol
\cite{andrews2004integer,olsson1994combinatorics,andrews1998theory} corresponds to the zero index case where $q=p$.
If $M$ is a Maya diagram, then for any $k\in \Z$ so is
\[ M+k = \{ m+k \colon m\in M \}.\]
The behaviour of the index $s_M$ under translation of $k$ is given by
\begin{equation}\label{eq:indexshift}
M'=M+k\quad \Rightarrow \quad s_{M'}=s_M+k.
\end{equation}
We will refer to an equivalence class of Maya diagrams related by such
shifts as an \textit{unlabelled Maya diagram}. One can visualize the
passage from an unlabelled to a labelled Maya diagram as the choice of
placement of the origin.

A Maya diagram $M\subset \Z$ is said to be in standard form if $p=0$
and $t_q>0$.  Visually, a Maya diagram in standard form has only
filled boxes $\boxdot$ to the left of the origin and one empty box
$\emptybox$ just to the right of the origin. Every unlabelled Maya
diagram permits a unique placement of the origin so as to obtain a
Maya diagram in standard form.

\begin{shaded}
\begin{exercise}
Draw the box-and-ball representation of the Maya diagram
\[M=\{\dots,-9,-8,-7,-5,-4,-1,1,2\}.   \]
Find the Frobenius symbol and the index of $M$. Find a translation $k$ such that $M'=M+k$ is in standard form, and write the Frobenius symbol, index and  box-and-ball representation of $M'$.
\end{exercise}
\end{shaded}

The solution to the previous exercise can be found in the figure below.
\begin{figure}[ht]
  \includegraphics[width=0.95\textwidth]{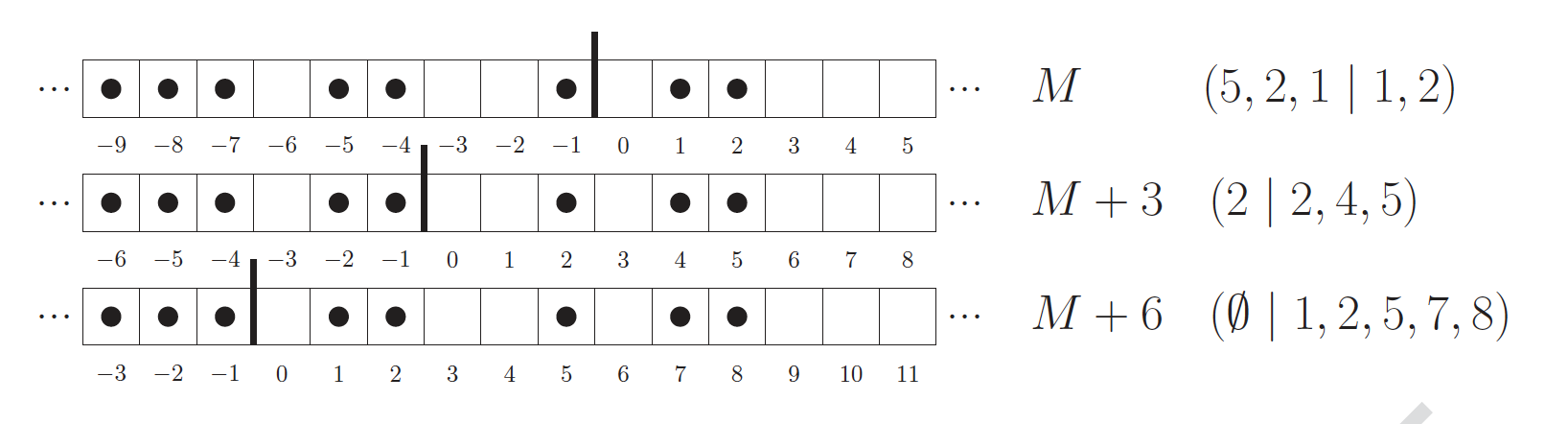}
\caption{Three equivalent Maya diagrams corresponding to the partition
$\lambda=(4,4,3,1,1)$, together with their Frobenius representation.}
  \label{fig:equivM}
\end{figure}

Observe that the third diagram is in standard form, so $k=6$ is the necesary shift. 
\subsection{Hermite pseudo-Wronskians}

We can interpret a Maya diagram with Frobenius symbol $(s_1,\dots,s_r|t_q,\dots,t_1)$ as the multi-index that specifies a multi-step rational Darboux transformation on the harmonic oscillator, i.e. $L\mapsto L_M$, where
\[L_M=-D_{xx} + x^2 -2\left( \Wr[\varphi_{-s_1},\dots,\varphi_{-s_r},\varphi_{t_1},\dots,\varphi_{t_q}] \right)_{xx} \]
where $\varphi_k$ are the seed functions for rational Darboux transformations of the harmonic oscillator described in \eqref{eq:seedHO}.
The first tuple in the Frobenius symbol specifies seed functions with conjugate Hermite polynomials in \eqref{eq:seedHO} (virtual states) while the second tuple 
specifies the bound states in \eqref{eq:seedHO}.
Getting rid of an overall exponential factor, we can associate to every Maya diagram a polynomial called a Hermite pseudo-Wronskian.
\begin{defn}
  Let $M$ be a Maya diagram and $(s_1,\dots,s_r|t_q,\dots,t_1)$ its
  corresponding Frobenius symbol. Define the polynomial
  \begin{equation}\label{eq:pWdef1} H_M = e^{-rx^2}\Wr[ e^{x^2}
    \th_{s_1},\ldots, e^{x^2} \th_{s_r}, H_{t_q},\ldots H_{t_1} ],
  \end{equation} where $\Wr$ denotes the Wronskian determinant of the
  indicated functions, and
  \begin{equation}
    \label{eq:thndef}
    \th_n(x)={\rm i}^{-n} H_{n}({\rm i}x)
  \end{equation}
  is the $n\supth$ degree conjugate Hermite polynomial.
\end{defn}

It is not evident that $H_M$ in \eqref{eq:pWdef1} is a polynomial, but this becomes clear once we represent it using a slightly different determinant.

\begin{prop}\label{prop:HM} The Wronskian $H_M$ admits the following alternative
  determinantal representation
  \begin{equation}\label{eq:pWdef2} H_M =
    \begin{vmatrix} \th_{s_1} & \th_{s_1+1} & \ldots &
\th_{s_1+r+q-1}\\ \vdots & \vdots & \ddots & \vdots\\ \th_{s_r} &
\th_{s_r+1} & \ldots & \th_{s_r+r+q-1}\\ H_{t_q} & D_x H_{t_q} &
\ldots & D_x^{r+q-1}H_{t_q}\\ \vdots & \vdots & \ddots & \vdots\\
H_{t_1} & D_x H_{t_1} & \ldots & D_x^{r+q-1}H_{t_1}
    \end{vmatrix}
  \end{equation}

\end{prop} 
\noindent
The term Hermite
pseudo-Wronskian was coined in \cite{gomez2016durfee} because  \eqref{eq:pWdef2} is a mix of a Casoratian and a Wronskian
determinant.

\begin{exercise}
\begin{shaded}
Prove Proposition \ref{prop:HM}, i.e. prove the relation 
\[  H_M = e^{-rx^2}\Wr[ e^{x^2}
    \th_{s_1},\ldots, e^{x^2} \th_{s_r}, H_{t_q},\ldots H_{t_1} ]=  \begin{vmatrix} \th_{s_1} & \th_{s_1+1} & \ldots &
\th_{s_1+r+q-1}\\ \vdots & \vdots & \ddots & \vdots\\ \th_{s_r} &
\th_{s_r+1} & \ldots & \th_{s_r+r+q-1}\\ H_{t_q} & D_x H_{t_q} &
\ldots & D_x^{r+q-1}H_{t_q}\\ \vdots & \vdots & \ddots & \vdots\\
H_{t_1} & D_x H_{t_1} & \ldots & D_x^{r+q-1}H_{t_1}
    \end{vmatrix}\]
    \vskip0.2cm
\end{shaded}
\end{exercise}

 The desired identity follows by the fundamental relations satified by Hermite polynomials
  \begin{equation}
    \label{eq:hermids}
    \begin{aligned}    \noindent
      &D_x H_n(x) = 2n H_{n-1}(x),\quad n\geq 0,\\     \noindent
      &D_x \th_n(x) = 2n \th_{n-1}(x),\quad n\geq 0,\\     \noindent
      &2x H_n(x) = H_{n+1}(x) + 2n H_{n-1}(x),\\     \noindent
      &2x \th_n(x) = \th_{n+1}(x) - 2n \th_{n-1}(x),\\
      &D_x (e^{x^2} \th_n(x)) 
        = e^{x^2}\th_{n+1}(x),\\
      &D_x (e^{-x^2} h_n(x)) = -e^{-x^2} h_{n+1}(x).
    \end{aligned}
  \end{equation}
  together with the Wronskian identity
  \begin{equation}
    \label{eq:Wrhomog}
    \Wr[g f_1,\ldots, g f_s] = g^s \Wr[f_1,\ldots, f_s], 
  \end{equation}

One remarkable property satisfied by all Maya diagrams in the same equivalence class, is that their associated Hermite pseudo-Wronskians enjoy a very simple relation: with an appropriate scaling, the Hermite pseudo-Wronskian of a given Maya diagram is invariant under translations.

\begin{prop}\label{prop:equiv}
  Let  $\hat{H}_M $ be the normalized pseudo-Wronskian
  \begin{equation}
    \label{eq:hHdef}
    \hat{H}_M = \frac{(-1)^{rq}H_M}{\prod_{1\leq i<j\leq r} (2s_j-2s_i)\prod_{1\leq
        i<j\leq q}
      (2 t_i-2t_j)}.
  \end{equation}
Then for any Maya diagram $M$ and $k\in\Z$ we have
  \begin{equation} \label{eq:HMequiv}
       \hat{H}_M =  \hat{H}_{M+k}.
  \end{equation}
\end{prop}

The proof of this Proposition is not too hard and proceeds by induction: it is enough to prove the equality by a shift of $k=1$. We leave it as an exercise for the interested reader. The proof can be seen in \cite{gomez2016durfee}. At least, to gain some practice and convince ourselves of this result, we propose the following exercise.

\begin{shaded}
\begin{exercise}
Let $M$ be the Maya diagram with Frobenius symbol $(3,2|2,4)$. Write down $\hat H_M$, $\hat H_{M+3}$ and $\hat H_{M-4}$. Compute the determinants and check that \eqref{eq:HMequiv} is verified.
\end{exercise}
\end{shaded}

The remarkable aspect of equation \eqref{eq:HMequiv} is that the identity involves determinants of different sizes. As mentioned above, every unlabelled Maya diagram contains a Maya diagram in standard form, and its associated Hermite pseudo-Wronskian \eqref{eq:pWdef1} is just an ordinary Wronskian determinant whose entries are Hermite polynomials.An interesting problem is to determine the smallest determinant in a given equivalence class, i.e. the minimum number of Darboux transformations to reach a givne potential. The details on how to solve this problem are given in \cite{gomez2016durfee}. 

Due to Proposition \ref{prop:equiv}, we could restrict the analysis without loss of generality to Maya diagrams in standard form and Wronskians of Hermite polynomials, but we will employ the general notation as it 
brings conceptual clarity to the description of Maya cycles.

We will now introduce and study a class of potentials for
Schr\"odinger operators that will be used as building blocks for
cyclic dressing chains: the set of rational extensions of the harmonic oscillator, which, as we will see, amounts to the set of potentials that one can obtain from $U(x)=x^2$ by applying rational Darboux-Crum transformations.

\subsection{Rational extensions of the harmonic oscillator}

\begin{defn}
A rational extension of the harmonic
oscillator is a potential of the form
\[ U(x) = x^2 + \frac{a(x)}{b(x)},\qquad a,b \text{
  polynomials},\qquad \deg a\leq \deg b, \]
that is \textit{exactly solvable by polynomials}, in the sense of Definition~\ref{def:ESP}.
\end{defn}

If $b(x)$ has no real zeros, then $L$ is a Sturm-Liouville operator on $\R$ with quasi-polynomial eigenfunctions. The next Proposition proved in \cite{gomez2013rational} states
that rational extensions of the harmonic oscillator can be put in one
to one correspondence with Maya diagrams. The details of this result are based on the theory of trivial monodromy potentials and they exceed the scope of these lecture notes. The interested reader is referred to \cite{gomez2013rational} and \cite{oblomkov1999monodromy} for further details.

\begin{prop}
  \label{prop:ratext}
Let $M\subset \Z$ be a Maya diagram.  Define
\begin{equation}
  \label{eq:UMdef}
  U_M(x) = x^2 - 2 D_x^2 \log H_M + 2s_M,
\end{equation}
where $H_M$ is the corresponding pseudo-Wronskian \eqref{eq:pWdef1}-
\eqref{eq:pWdef2}, and $s_M\in \Z$ is the index of $M$.
Up to an additive constant, every rational extension of the harmonic
  oscillator takes the form \eqref{eq:UMdef}.
\end{prop}

The class of Schr\"odinger operators with potentials that are rational extensions of the harmonic oscillator is invariant
under a rational Darboux transformations. Otherwise speaking, if we perform a rational Darboux transformation on a rational extension of the harmonic oscillator, indexed by a Maya diagram $M$, we will obtain another potential in the same class, indexed by $M'$. Both $M$ and $M'$ differ only in one element, as we show next.

\begin{defn}
We define the flip at position $m\in \Z$ to be the involution
$\phi_m:\cM\to \cM$ defined by
\begin{equation}\label{eq:flipdef}
 \phi_m : M \mapsto
\begin{cases}
   M \cup \{ m \} & \text{ if } m\notin M \\
   M \setminus \{ m \} & \text{ if } m\in M 
\end{cases},\qquad M\in \cM.
\end{equation}
\end{defn}
\noindent
In the first case, we say that $\phi_m$ acts on $M$ by a
state-deleting transformation ($\emptybox\to \boxdot$).  In the second
case, we say that $\phi_m$ acts by a state-adding transformation
($\boxdot\to\emptybox$).

Using Crum's formula for iterated Darboux transformations \eqref{eq:Crum}, and the seed functions for rational DTs of the harmonic oscillator \eqref{eq:seedHO}, it can be shown that every quasi-rational eigenfunction of $ L=-D_x^2+ U_M(x)$
has the form
\begin{equation}
  \label{eq:seedfunc}
  \psi_{M,m} = e^{\epsilon x^2/2}\frac{H_{\phi_m(M)}}{H_M}, \qquad m\in \Z,
\end{equation}
with 
\[ \epsilon = \begin{cases}
  -1 & \text{ if } \, m\notin M \\
  +1 & \text{ if } \, m\in M 
\end{cases} \,.
\]
Explicitly, we have
\begin{equation}
    \label{eq:Mneigenfunc}
    L \psi_{M,m}  = (2m+1) \psi_{M,m} ,\quad m\in \Z.    
  \end{equation}

\begin{rem}
The seed eigenfunctions \eqref{eq:seedfunc} include the true eigenfunctions of $L$ plus other set of formal non square-integrable eigenfunctions, sometimes known in the physics literature as \textit{virtual states},\cite{odake2013krein,odake2011exactly}. For a correct spectral theoretic interpretation one needs to ensure that the potential $U_M$ is regular, i.e. that $H_M$ has no zeros in $\R$. The set of Maya diagrams for which $H_M$ has no real zeros was characterized (in a more general setting) independently by Krein \cite{krein1957continuous} and Adler \cite{adler1994modification}, while the number of real zeros for $H_M$ was given in \cite{garcia2015oscillation}. However, for the purpose of this paper it is convenient stay within a purely formal setting and keep the whole class of potentials $U_M$, regardless of whether they have real poles or not.
\end{rem}

The relation between dressing chains of Darboux transformations for the class of operators \eqref{eq:UMdef} and flip operations on Maya diagrams is made explicit by the following proposition. 

\begin{prop}
  \label{prop:UMflip}
  Two Maya diagrams $M, M'$ are related by a flip \eqref{eq:flipdef}
  if and only if their associated rational extensions $U_M,U_{M'}$ are
  connected by a Darboux transformation \eqref{eq:Uplus1}.
\end{prop}
\begin{proof}
  Suppose that $m\notin M$ and that $M' = M \cup \{ m\}$ is a
  state-deleting flip transformation of $M$. The seed function for the factorization is $\psi_{M,m}$ defined in \eqref{eq:seedfunc}.   Set
  \begin{equation}
    \label{eq:f1statedelete}
    w_{M,m} = \frac{\psi_{M,m}'}{\psi_{M,m}} = -x +
    \frac{H_{M'}'}{H_{M'}} - \frac{H_{M}'}{H_{M}}.
  \end{equation}
  Since
  \[ s_{M'} = s_M +1,\]
  by \eqref{eq:UMdef}, we have 
  \begin{equation}
    \label{eq:UMM'}
    \frac12(U_{M'}-U_{M}) = 1+ D_x \left(
      \frac{H_{M}'}{H_{M}}-  \frac{H_{M'}'}{H_{M'}} \right) = -w_{M,m}',
  \end{equation}
so that \eqref{eq:Uplus1} holds.
  Conversely, suppose that $M$ and $M'$ are such that \eqref{eq:UMM'}
  holds for some $w = w(x)$. If we define
  \[ w = \frac{\psi'}{\psi},\qquad
   \psi = e^{-x^2/2}\frac{ H_{M'}}{H_{M}} ,\]
   then $\psi $ must be a quasi-rational seed function for $U_M$ and
   it follows by \eqref{eq:seedfunc} of Proposition \ref{prop:ratext} that
   that $M'= M\cup \{m\}$ for some $m\notin M$. The corresponding result for state-adding Darboux transformations is done in a similar way.
\end{proof}

We see thus that the class of rational extensions of the harmonic
oscillator is indexed by Maya diagrams, and that the Darboux
transformations that preserve this class can be described by flip
operations on Maya diagrams. It has recently been noticed that Maya diagrams and rational extensions can be realized as categories, and their relation as a functor between categories, \cite{GGMM2019}.
Now we are ready to introduce the concept of cyclic Maya diagrams, and use them later to build Darboux dressing chains on these potentials, and solutions to $A_N^{(1)}$-Painlev\'e.

\subsection{Cyclic Maya diagrams}

Cyclic Maya diagrams are just the ones such that we can perform a number of flip operations on them, and recover the same Maya diagram up to a shift, \cite{ggsm}. We introduce the necessary notation and precise definitions  below.

\begin{defn}
  \label{def:multiflip}
  For $p\in \N$ let $\cZ_p$ denote the set of all subsets of $\Z$
  having cardinality $p$.
  For $\bmu= \{ \mu_1,\ldots, \mu_p\}\in \cZ_p$ we now define
  $\phi_{\bmu}$ to be the multi-flip
   \begin{equation}
     \label{eq:phimudef}
     \phi_{\bmu}= \phi_{\mu_1} \circ \cdots \circ \phi_{\mu_{p}}.
\end{equation}
\end{defn}

\begin{defn}
  We say that $M$ is $p$-cyclic with shift $k$, or $(p,k)$ cyclic, if
  there exists a $\bmu \in \cZ_p$ such that
  \begin{equation}
    \label{eq:cyclicMdef}
    \phi_\bmu(M) = M+k.
  \end{equation}
  We will say that $M$ is $p$-cyclic if it is $(p,k)$ cyclic for some
  $k\in \Z$.
\end{defn}

\begin{prop}
  \label{prop:muM1M2}
  For Maya diagrams $M,M'\in \cM$, define the set
  \begin{equation}
    \label{eq:muM1M2}
    \Upsilon(M,M') = (M \setminus M') \cup (M'\setminus M)
  \end{equation}
  Then the multi-flip $\phi_\bmu$ where $\bmu=\Upsilon(M,M')$ is the unique
  multi-flip such that $ M' = \phi_{\bmu}(M)$ and $\Upsilon(M,M') $ 
  $M = \phi_{\bmu}(M')$. 
\end{prop}
Intuitively, is the set of sites at which $M$ and $M'$ differ, so it is evident that a multi-flip on these sites will turn $M$ into $M'$ and viceversa.
As an immediate corollary, we have the following.
\begin{prop}
  Let $k$ be a non-zero integer.  Every Maya diagram $M\in \cM$ is
  $(p,k)$ cyclic where $p$ is the cardinality of
  $\bmu=\Upsilon(M,M+k)$.
\end{prop}


\begin{shaded}
\begin{exercise}
For the following Maya diagrams, find the sequence of flip transformations $\bmu = \{\mu_0,\mu_1,\mu_2\}$ such that $M'=\phi_{\bmu}(M)=M+k$
\begin{align*}
&k=1 \qquad M=(\emptyset | 3,4,5,6) =(-\infty,-1] \cup \{3,4,5,6\}\\
&k=3 \qquad  M=(3 | 1,2,4,5,8) =(-\infty,-4]\cup\{ -2,-1\} \cup \{ 1,2,4,5,8\}
\end{align*}
\end{exercise}
\end{shaded}
In the first case, we see that $M'=M+1=(\emptyset|0,4,5,6,7 )$, so $\bmu = \Upsilon(M,M+1)=(0,3,7)$. The first and third flips correspond to state-deleting transformations ($\emptybox\to\boxdot$), while the second is a state-adding transformation ($\boxdot\to\emptybox$).
In the second case, we have
 \[M'=M+3=(\emptyset|1,2,4,5,7,8,11)=(-\infty,-1]\cup\{1,2,4,5,7,8,11 \}\]
  so $\bmu = \Upsilon(M,M+3)=(-3,7,11 )$. In this case, all three transformations are state-deleting ($\emptybox\to\boxdot$).

Now we are  able to establish the link between Maya cycles and cyclic
dressing chains composed of rational extensions of the harmonic
oscillator.
\begin{thm}
  \label{prop:Mwcorrespondence}
  Let $M\in \cM$ be a Maya diagram, $k$ a non-zero integer, and $p$
  the cardinality of $\bmu=\Upsilon(M,M+k)$. Let
  $\bmu = \{\mu_0,\ldots, \mu_{p-1}\}$ be an arbitrary enumeration of
  $\bmu$ and set
  \begin{equation}
    \label{eq:Mchain}
    M_0 = M,\quad M_{i+1} = \phi_{\mu_i}(M_{i}),\qquad
    i=0,1,\ldots, p-1
  \end{equation}
  so that $M_p = M_0+k$ by construction.  Set
  \begin{align}
    &w_i= s_{i} \,x+ \frac{H_{M_{i+1}}'}{H_{M_{i+1}}}- \frac{H_{M_{i}}'}{H_{M_{i}}},\qquad i=0,\dots,p-1. \label{eq:HM2w}\\
    &\alpha_i=2(\mu_i-\mu_{i+1}), \label{eq:mu2alpha}
      \intertext{where}
     \label{eq:sign}
    &s_i=\begin{cases}
      -1 & \textit{ if } \, \mu_i\notin M \\
    +1 & \textit{ if } \, \mu_i\in M 
  \end{cases},
\end{align}
and 
\[ \mu_p = \mu_0 + k.\]
Then, $(w_0,\ldots, w_{p-1}; \alpha_0,\ldots, \alpha_{p-1})$
constitutes a rational solution to the $p$-cyclic dressing chain
\eqref{eq:wfchain} with shift $\Delta=2k$.
\end{thm}

\begin{proof}
The result follows from the structure of the seed eigenfunctions \eqref{eq:seedfunc} with eigenvalues given by \eqref{eq:Mneigenfunc}, after applying \eqref{eq:Lipsii} and \eqref{eq:alphaidef}.
The sign of $s_{i}$ indicates whether the $(i+1)$-th step of the chain that takes $L_i$ to $L_{i+1}$ is a state-adding $(+1)$ or state-deleting $(-1)$ transformation.
\end{proof}

So now we know that given a Maya $n$-cycle, we can build an $n$-cyclic dressing chain and a rational solution to the Noumi-Yamada system. But we would like to go further and classify  cyclic Maya diagrams for any given (odd) period, which we
tackle next. 
\begin{rem}
 Under the correspondence described by Proposition
\ref{prop:Mwcorrespondence}, the reversal symmetry
\eqref{eq:reversal2} manifests as the transformation
\[ (M_0,\ldots, M_p) \mapsto (M_p,\ldots, M_0),\quad (\mu_1,\ldots,
\mu_{p}) \mapsto (\mu_{p},\ldots, \mu_1),\quad k\mapsto -k.\]
In light of the above remark, there is no loss of generality if we
restrict our attention to cyclic Maya diagrams with a positive shift
$k>0$.
\end{rem}

\section{Classification of cyclic Maya diagrams}\label{sec:Mayacycles}

In this section we introduce two new concepts on Maya diagrams:  \textit{genus} and
\textit{interlacing}, which become a key ingredient in the characterization of cyclic Maya
diagrams. But before we do so, let us introduce another way to specify a Maya diagram, which becomes more convenient for the task that we now face.



For $\bbeta\in \cZ_{2g+1}$ define the Maya diagram
\begin{equation}
  \label{eq:MBi}
  \Xi(\bbeta)= (-\infty,\beta_0) \cup [\beta_1,\beta_2) \cup
  \ \cdots \cup [\beta_{2g-1},\beta_{2g})
\end{equation}
where
\[ [m,n) = \{ j\in \Z \colon m\leq j < n\}\]
and where $\beta_0<\beta_1<\cdots < \beta_{2g}$ is the strictly increasing
enumeration of $\bbeta$.

\begin{prop}
  Every Maya diagram $M\in \cM$ has a unique representation of the
  form $M=\Xi(\bbeta)$ where $\bbeta$ is a set of integers of odd
  cardinality $2g+1$.
\end{prop}
\begin{defn}\label{def:genus}
  We call the integer $g\geq 0$ the genus of $M= \Xi(\bbeta)$ and
  $(\beta_0,\beta_1,\ldots, \beta_{2g})$ the block coordinates of $M$.
\end{defn}
\noindent

\begin{rem}
  To motivate Definition \ref{def:genus}, it is perhaps more illustrative to
  understand the visual meaning of the genus of $M$, see
the Maya diagram below.  After removal of the initial infinite
  $\boxdot$ segment and the trailing infinite $\emptybox$ segment, a
  Maya diagram consists of alternating empty $\emptybox$ and filled
  $\boxdot$ segments of variable length.  The genus $g$ counts the
  number of such pairs.  
  The even block coordinates $\beta_{2i}$ indicate the starting
  positions of the empty segments, and the odd block coordinates
  $\beta_{2i+1}$ indicated the starting positions of the filled
  segments.  Also, note that $M$ is in standard form if and only if
  $\beta_0=0$.
\end{rem}

\begin{shaded}
\begin{exercise}
Draw the box-ball diagram corresponding to the genus-$2$ Maya diagram with block coordinates $(\beta_0,\ldots, \beta_4) = (2,3,5,7,10)$ and give its Frobenius symbol.
\end{exercise}
\end{shaded}

Since $\beta_0=2$, to the left of $2$ all sites are filled and site $2$ is empty. Next we have  filled block $[\beta_1,\beta_2)=[3,5)$ of size $2$ and another filled block at  $[\beta_3,\beta_3)=[7,10)$ of size $3$. All sites are empty to the right of $\beta_4=10$.

\begin{tikzpicture}[scale=0.6]

\draw  (1,1) grid +(15 ,1);

\path [fill] (0.5,1.5) node {\huge ...} 
++(1,0) circle (5pt) ++(1,0) circle (5pt)  ++(1,0) circle (5pt) 
++(1,0) circle (5pt) ++(1,0) circle (5pt) 
++(2,0) circle (5pt) ++(1,0) circle (5pt) 
++ (3,0) circle (5pt)  ++(1,0) circle (5pt)   ++ (1,0) circle (5pt) 
++ (3,0) node {\huge ...} +(1,0) node[anchor=west] {};

\draw[line width=1pt] (4,1) -- ++ (0,1.5);

\foreach \x in {-3,...,11} 	\draw (\x+4.5,2.5)  node {$\x$};
\path (6.5,0.5) node {$\beta_0$} ++ (1,0) node {$\beta_1$}
++ (2,0) node {$\beta_2$}++ (2,0) node {$\beta_3$}++ (3,0) node {$\beta_4$}
;
\end{tikzpicture}

 $$M =  (-\infty,\beta_0)\cup [ \beta_1,\beta_2) \cup [ \beta_3,\beta_4)$$
Note that the genus is both the  number of finite-size empty blocks and the number of finite-size  filled blocks.
  

\begin{shaded}
\begin{exercise}\label{ex:+1}
  Let $M = \Xi(\bbeta)$ be a Maya diagram specified by its block
  coordinates  \label{eq:MBi}. Prove that 
    \[ \bbeta = \Upsilon(M,M+1).\]
\end{exercise}
\end{shaded}

\begin{proof}
  Observe that
  \[ M+1 = (-\infty, \beta_0] \cup (\beta_1,\beta_2] \cup \cdots \cup
  (\beta_{2g-1}, \beta_{2g}],\]
  where
  \[ (m,n] = \{ j\in \Z \colon m<j\leq n \}.\]
  It follows that
  \begin{align*}
    (M+1)\setminus M &= \{ \beta_0, \ldots, \beta_{2g} \}\\
    M\setminus (M+1) &= \{ \beta_1,\ldots, \beta_{2g-1} \}.
  \end{align*}
  The desired conclusion follows immediately.
\end{proof}


Let $\cM_g$ denote the set of Maya diagrams of genus $g$.  The above
discussion may be summarized by saying that the mapping \eqref{eq:MBi}
defines a bijection $\Xi:\cZ_{2g+1} \to \cM_g$, and that the block
coordinates are precisely the flip sites required for a translation
$M\mapsto M+1$.

The next concept we need to introduce is the interlacing and modular
decomposition. 
\begin{defn}\label{def:interlacing}
  Fix a $k\in \N$ and let $M^{(0)}, M^{(1)},\ldots M^{(k-1)}\subset \Z$ be sets
  of integers.  We define the interlacing of these to be the set
  \begin{equation}\label{eq:interlacing} \Theta\left(M^{(0)}, M^{(1)},\ldots M^{(k-1)}\right)
    = \bigcup_{i=0}^{k-1} (k M^{(i)} +i),
 \end{equation}
 where
 \[ kM +j = \{ km + j \colon m\in M \},\quad M\subset \Z.\]
 Dually, given a set of integers $M\subset \Z$ and a $k\in \N$ define
 the sets
 \[ M^{(i)} = \{ m\in \Z \colon km+i \in M\},\quad i=0,1,\ldots, k-1.\]
 We will call the $k$-tuple of sets $\left(M^{(0)}, M^{(1)},\ldots M^{(k-1)}\right)$ the
 $k$-modular decomposition of $M$.
\end{defn}


The following result follows directly from the above definitions.
\begin{prop}
  We have $M=\Theta\left(M^{(0)}, M^{(1)},\ldots M^{(k-1)}\right)$ if and only if\\
  $\left(M^{(0)}, M^{(1)},\ldots M^{(k-1)}\right)$ is the $k$-modular decomposition of $M$.
\end{prop}

Even though the above operations of interlacing and modular
decomposition apply to general sets, they have a well defined
restriction to Maya diagrams.  Indeed, it is not hard to check that if
$M=\Theta\left(M^{(0)}, M^{(1)},\ldots M^{(k-1)}\right)$ and $M$ is a Maya diagram, then
 $M^{(0)}, M^{(1)},\ldots M^{(k-1)}$ are also Maya diagrams.  Conversely, if the latter are all Maya
diagrams, then so is $M$.  Another important case concerns the
interlacing of finite sets.  The definition \eqref{eq:interlacing}
implies directly that if $\bmu^{(i)} \in \cZ_{p_i},\; i=0,1,\ldots, k-1$
then
\[ \bmu = \Theta\left(\bmu^{(0)}, \ldots , \bmu^{(k-1)}\right) \]
is a finite set of cardinality $p=p_0+\cdots + p_{k-1}$.

Visually, each of the $k$ Maya diagrams is dilated by a factor of $k$,
shifted by one unit with respect to the previous one and superimposed,
so the interlaced Maya diagram incorporates the information from
$M^{(0)}, \ldots M^{(k-1)}$ in $k$ different modular classes. 
 In other words, the interlaced Maya diagram is built by copying sequentially a filled or
empty box as determined by each of the $k$ Maya diagrams.

\begin{shaded}
\begin{exercise}
 For the following three Maya diagrams, given by their block coordinates:
\[ M_0 = \Xi(0,1,4),\quad M_1 = \Xi(-1,1,3,5,6),\quad M_2 = \Xi(4)  \]
 Draw the box-square diagram of the interlaced diagram $M= \Theta(M_0,M_1,M_2)$ and give the block coordinates and the $3$-block coordinates of $M$. 
\end{exercise}
\end{shaded}

\begin{figure}[h]
\begin{tikzpicture}[scale=0.6]

\draw  (1,3) grid +(11 ,1);

\path [fill,color=black] (0.5,3.5) node {\huge ...} 
++(1,0) circle (5pt) ++(1,0) circle (5pt)  ++(1,0) circle (5pt) 
++(1,0) circle (5pt)
++(2,0) circle (5pt) ++(1,0) circle (5pt)  ++(1,0) circle (5pt) 
++ (4,0) node {\huge ...} +(1,0) node[anchor=west,color=black] { $M_0 = \Xi(0,1,4),\;\qquad\quad\,\,\,\, g_0 = 1$}; 

\draw[line width=1pt] (5,3) -- ++ (0,2);

\foreach \x in {-4,...,6} 	\draw (\x+5.5,4.5)  node {$\x$};

\draw  (1,1) grid +(11 ,1);

\path [fill,color=blue] (0.5,1.5) node {\huge ...} 
++(1,0) circle (5pt) ++(1,0) circle (5pt)  ++(1,0) circle (5pt) 
++(3,0) circle (5pt) ++(1,0) circle (5pt) 
++ (3,0) circle (5pt)  ++(2,0) node {\huge ...} +(1,0) node[anchor=west,color=black] { $M_1 = \Xi(-1,1,3,5,6),\;\quad g_1 = 2$}; 

\draw[line width=1pt] (5,1) -- ++ (0,2);

\draw  (1,-1) grid +(11 ,1);

\path [fill,color=red] (0.5,-0.5) node {\huge ...} 
++(1,0) circle (5pt) ++(1,0) circle (5pt)  ++(1,0) circle (5pt) 
++(1,0) circle (5pt) ++(1,0) circle (5pt)  ++(1,0) circle (5pt) 
++(1,0) circle (5pt)  ++(1,0) circle (5pt) 
++(4,0) node {\huge ...} +(1,0) node[anchor=west,color=black] { $M_2 = \Xi(4),\qquad\qquad\qquad g_2 = 0$}; 

\draw[line width=1pt] (5,-1) -- ++ (0,2);

\draw  (0,-4) grid +(23 ,1);
\foreach \x in {-5,...,17} 	\draw (\x+5.5,2.5-5)  node {$\x$};
\draw[line width=1pt] (5,-4) -- ++ (0,2);

\path [fill,color=black] (2.5,-3.5)   
circle (5pt) ++(6,0) circle (5pt)  
++(3,0) circle (5pt) ++(3,0) circle (5pt) ;

\path [fill,color=blue] 
(0.5,-3.5) circle (5pt) ++(9,0) circle (5pt)  
++(3,0) circle (5pt) ++(9,0) circle (5pt) ;

\path [fill,color=red] 
(1.5,-3.5) circle (5pt) ++(3,0) circle (5pt)  
++(3,0) circle (5pt) ++(3,0) circle (5pt)
++(3,0) circle (5pt) ++(3,0) circle (5pt) ;

\draw (1.0,-5) node[right] {$M= \Theta(M_0,M_1,M_2)=\Xi_3(0,1,4|-1,1,3,5,6|4) = \Xi(-2,-1,0,2,10,11,12,16,17)$}; 
\end{tikzpicture}

\end{figure}

Equipped with these notions of genus and interlacing, we are now ready
to state the main result for the classification of cyclic Maya
diagrams.


\begin{thm}
  \label{thm:Mp}
  Let $M=\Theta\left(M^{(0)}, M^{(1)},\ldots M^{(k-1)}\right)$ be the $k$-modular decomposition of a given Maya diagram $M$.  Let $g_i$ be the genus
  of $M^{(i)},\; i=0,1,\ldots, k-1$.  Then, $M$ is $(p,k)$-cyclic where
  \begin{equation}
    \label{eq:pgi}
    p = p_0+p_1+\cdots + p_{k-1},\qquad p_i = 2g_i + 1.
  \end{equation}
\end{thm}
\begin{proof}
  Let $\bbeta^{(i)}=\Upsilon\left(M^{(i)},M^{(i+1)}\right) \in \cZ_{p_i}$ be the block
  coordinates of $M^{(i)},\; i=0,1,\ldots, k-1$.  Consider the interlacing
  $\bmu = \Theta\left(\bbeta^{(0)}, \ldots , \bbeta^{(k-1)}\right)$.  From
  Exercise~\ref{ex:+1} we have that,
  \[ \phi_{\bbeta^{(i)}}\left( M^{(i)}\right)= M^{(i)} +1.\]
  so it follows that
  \begin{align*}
    \phi_{\bmu}(M)&= \phi_{\Theta\left(\bbeta^{(0)}, \ldots , \bbeta^{(k-1)}\right)}\Theta\left( M^{(0)} , \ldots ,
                  M^{(k-1)}\right) \\   
    &= \Theta\left( \phi_{\bbeta^{(0)}}(M^{(0)}) , \ldots ,
                    \phi_{\bbeta^{(k-1)}}(M^{(k-1)})\right) \\
                  &= \Theta\left(M^{(0)}+1 , \ldots , M^{(k-1)} + 1\right) \\
                  &= \Theta\left(M^{(0)}, \ldots , M^{(k-1)}\right) + k \\
                  &= M+k.
  \end{align*}
  Therefore, $M$ is $(p,k)$ cyclic where the value of $p$ agrees with
  \eqref{eq:pgi}.
\end{proof}

Normally we will have a bound on the period $p$, and the classification problem is to describe all possible $(p,k)$-cyclic Maya diagrams for all values of $k>0$.
Theorem~\ref{thm:Mp} sets the way to do this.
\begin{cor}\label{cor:k}
  For a fixed period $p\in \N$, there exist $p$-cyclic Maya diagrams with
  shifts $k=p,p-2,\dots,\lfloor p/2\rfloor$, and no other positive shifts are  possible.
\end{cor}

\begin{rem}
  The highest shift $k=p$ corresponds to the interlacing of $p$ trivial (genus
  0) Maya diagrams.
\end{rem}

\subsection{Indexing Maya $p$-cycles}

We now introduce a combinatorial system for describing rational
solutions of $p$-cyclic factorization chains.  First, we require a
 generalized notion of block coordinates suitable for
describing $p$-cyclic Maya diagrams.
\begin{defn}\label{def:kblock}
  Let $M=\Theta\left(M^{(0)},\ldots M^{(k-1)}\right)$ be a $k$-modular decomposition of a
  $(p,k)$ cyclic Maya diagram.  For $i=0,1,\ldots, k-1$ let
  $\bbeta^{(i)}= \left(\beta^{(i)}_{0}, \ldots,
    \beta^{(i)}_{p_i-1}\right)$
  be the block coordinates of $M^{(i)}$ enumerated in increasing order.
  In light of the fact that
  \[ M = \Theta\left(\Xi(\bbeta^{(0)}), \ldots , \Xi(\bbeta^{(k-1)})\right),\]
  we will refer to the concatenated sequence
  \begin{align*}
    (\beta_0,\beta_1,\ldots, \beta_{p-1}) 
    &= (\bbeta^{(0)} | \bbeta^{(1)} | \ldots |
      \bbeta^{(k-1)}) \\
    &=  \left( \beta^{(0)}_{0}, \ldots, \beta^{(0)}_{p_0-1} |
      \beta^{(1)}_{0}, \ldots, \beta^{(1)}_{p_1-1} 
      | \ldots | \beta^{(k-1)}_{0},\ldots, \beta^{(k-1)}_{p_{k-1}-1}\right) 
  \end{align*}
  as the $k$-block coordinates of $M$.  Formally, the correspondence
  between $k$-block coordinates and Maya diagram is described by the
  mapping
  \[ \Xi_k\colon \cZ_{2g_0+1} \times \cdots \times \cZ_{2g_{k-1}+1} \to \cM \]
  with action
  \[ \Xi_k \colon (\bbeta^{(0)} | \bbeta^{(1)} | \ldots |
      \bbeta^{(k-1)})\mapsto \Theta\left(\Xi(\bbeta^{(0)}), \ldots ,
    \Xi(\bbeta^{(k-1)})\right) \]
\end{defn}

\begin{defn}
  Fix a $k\in \N$. For $m\in \Z$ let $[m]_k\in \{ 0,1,\ldots, k-1 \}$
  denote the residue class of $m$ modulo division by $k$.  For
  $m,n \in \Z$ say that $m \preccurlyeq_k n$ if and only if
  \[ [m]_k < [n]_k,\quad \text{ or } \quad [m]_k = [n]_k\; \text{ and
  } m\leq n.\]
  In this way, the transitive, reflexive relation $\preccurlyeq_k$
  forms a total order on $\Z$.
\end{defn}

\begin{prop}
  Let $M$ be a $(p,k)$ cyclic Maya diagram.  There exists a unique
  $p$-tuple of integers $(\mu_0,\ldots, \mu_{p-1})$ strictly ordered
  relative to $\preccurlyeq_k$ such that
  \begin{equation}
    \label{eq:bmuMk}
    \phi_\bmu(M) = M+k 
  \end{equation}
\end{prop}
\begin{proof}
  Let
  $(\beta_0,\ldots, \beta_{p-1}) =  (\bbeta^{(0)} | \bbeta^{(1)} | \ldots |   \bbeta^{(k-1)})$ be the $k$-block coordinates of $M$.
  Set \[ \bmu = \Theta\left(\bbeta^{(0)}, \ldots , \bbeta^{(k-1)}\right) \]
  so that \eqref{eq:bmuMk} holds by the proof to Theorem \ref{thm:Mp}.
  The desired enumeration of $\bmu$ is given by
  \[ (k \beta_0,\ldots, k \beta_{p-1}) + (0^{p_0}, 1^{p_1}, \ldots,
  (k-1)^{p_{k-1}}) \]
  where the exponents indicate repetition.  Explicitly,
  $(\mu_0,\ldots, \mu_{p-1})$ is given by
  \[ \left(k\beta^{(0)}_{0},\ldots, k\beta^{(0)}_{p_0-1}, k\beta^{(1)}_{0} + 1 ,\ldots,
  k\beta^{(1)}_{p_1-1} +1 , \ldots , k \beta^{(k-1)}_{0} + k-1, \ldots,
  k\beta^{(k-1)}_{p_{k-1}-1} + k-1 \right) .\]
\end{proof}
\begin{defn}
  In light of \eqref{eq:bmuMk} we will refer to the just defined tuple
  $(\mu_0,\mu_1,\ldots, \mu_{p-1})$ as the $k$-canonical flip sequence
  of $M$ and refer to the tuple $(p_0,p_1,\ldots, p_{k-1})$ as the
  $k$-signature of $M$.
\end{defn}

By Proposition \ref{prop:Mwcorrespondence} a rational solution of the
$p$-cyclic dressing chain requires a $(p,k)$ cyclic Maya diagram, and
an additional item data, namely a fixed ordering of the canonical
flip sequence.
We will specify such ordering  as
\[ \bmu_\bpi = (\mu_{\pi_0},\ldots, \mu_{\pi_{p-1}}) \]
where $\bpi=(\pi_0,\ldots, \pi_{p-1})$ is a permutation of
$(0,1,\ldots, p-1)$. With this notation, the chain of Maya diagrams
described in Proposition \ref{prop:Mwcorrespondence}
is generated as
\begin{equation}\label{eq:picycle}
 M_0 = M,\qquad M_{i+1} =  \phi_{\mu_{\pi_i}}(M_i),\qquad i=0,1,\ldots, p-1.
 \end{equation}
\begin{rem}\label{rem:normalization}
Using a translation it is possible to normalize $M$ so that
$\mu_{0} = 0$.  Using a cyclic permutation and it is possible to
normalize $\bpi$ so that $\pi_p=0$.  The net effect of these two
normalizations is to ensure  that
$M_0,M_1,\ldots, M_{p-1}$ have standard form.
\end{rem}

\begin{rem}
In the discussion so far we have imposed the hypothesis that the
sequence of flips that produces a translation $M\mapsto M+k$ does not
contain any repetitions.  However, in order to obtain a full
classification of rational solutions, it will be necessary to account
for degenerate chains which include multiple flips at the same site.

To that end it is necessary to modify Definition \ref{def:multiflip}
to allow $\bmu$ to be a multi-set\footnote{A multi-set is
  generalization of the concept of a set that allows for multiple
  instances for each of its elements.}, and to allow
$\mu_0,\mu_1,\ldots, \mu_{p-1}$ in \eqref{eq:MBi} to be merely a
non-decreasing sequence.  
This has the effect of permitting $\emptybox$ and $\boxdot$ segments of
zero length wherever $\mu_{i+1} = \mu_i$.  The $\Xi$-image of such a
non-decreasing sequence is not necessarily a Maya diagram of genus
$g$, but rather a Maya diagram whose genus is bounded above by $g$.

It is no longer possible to assert that there is a unique $\bmu$ such
that $\phi_\bmu(M) = M+k$, because it is possible to augment the
non-degenerate $\bmu = \Upsilon(M,M+k)$ with an arbitrary number of
pairs of flips at the same site to arrive at a degenerate $\bmu'$ such
that $\phi_{\bmu'}(M) = M+k$ also.  The rest of the theory remains
unchanged.
\end{rem}

\subsection{Rational solutions of $A_4$-Painlev\'e}\label{sec:A4}

In this section we will put together all the results derived above in order to describe an effective way of labelling and constructing all the rational solutions to the $A_{2k}$-Painlev\'e system based on cyclic dressing chains of rational extensions of the harmonic oscillator, \cite{cggm}. We conjecture that the construction described below covers all rational solutions to such systems. As an illustrative example, we describe all rational solutions to the $A_4$-Painlev\'e system, and we fournish examples in each signature class.

For odd $p$, in order to specify a Maya $p$-cycle, or equivalently a rational solution of a $p$-cyclic dressing chain, we need to specify three items of data:
\begin{enumerate}
\item a signature sequence $(p_0,\ldots, p_{k-1})$ consisting of odd positive integers that sum to $p$. This sequence determines the genus of the $k$ interlaced Maya diagrams that give rise to a $(p,k)$-cyclic Maya diagram $M$. The possible values of $k$ are given by Corollary~\ref{cor:k}.
\item Once the signature is fixed, we need to specify the $k$-block coordinates \[(\beta_0,\dots,\beta_{p-1})=(\bbeta^{(0)}|\ldots| \bbeta^{(k-1)})\] where $\bbeta^{(i)}=(\beta^{(i)}_0,\dots,\beta^{(i)}_{p_i})$ are the block coordinates that define each of the interlaced Maya diagrams $M^{(i)}$. These two items of data specify uniquely a $(p,k)$-cyclic Maya diagram $M$, and a canonical flip sequence $\bmu=(\beta_0,\dots,\beta_{p-1})$ . The next item specifies a given $p$-cycle that contains $M$.
\item Once the $k$-block coordinates and canonical flip sequence $\bmu$ are fixed, we still have the freedom to choose a permutation  $\bpi\in S_p$ of $(0,1,\ldots, p-1)$ that specifies the actual flip sequence $\bmu_\bpi$, i.e. the order in which the flips in the canonical flip sequence are applied to build the Maya $p$-cycle.

\end{enumerate}

For any signature of a Maya $p$-cycle, we need to specify the $p$ integers in the canonical flip sequence, but following Remark~\ref{rem:normalization}, we can get rid of translation invariance by setting $\mu_0=\beta^{(0)}_0=0$, leaving only $p-1$ free integers. Moreover, we can restrict ourselves to permutations such that $\bpi_p=0$ in order to remove the invariance under cyclic permutations. The remaining number of degrees of freedom is $p-1$, which (perhaps not surprisingly) coincides with the number of generators of the symmetry group $A^{(1)}_{p-1}$. This is a strong indication that the class described above captures a generic orbit of a seed solution under the action of the symmetry group.


We now illustrate the general theory by describing the rational
solutions of the $A^{(1)}_4$- Painlev\'e system, whose equations are given by

\begin{eqnarray}\label{eq:A4system}
f_0' + f_0(f_1-f_2+f_3-f_4) &=& \alpha_0, \nonumber\\
f_1' + f_1(f_2-f_3+f_4-f_0) &=& \alpha_1,\nonumber \\
f_2' + f_2(f_3-f_4+f_0-f_1) &=& \alpha_2, \\
f_3' + f_3(f_4-f_0+f_1-f_2) &=& \alpha_3, \nonumber \\
f_4' + f_4(f_0-f_1+f_2-f_3) &=& \alpha_4,\nonumber 
\end{eqnarray}
with normalization
\[ f_0+f_1+f_2+f_3+f_4=z.\]
\begin{thm}
  Rational solutions of the $A^{(1)}_4$-Painlev\'e system \eqref{eq:A4system} correspond to chains of $5$-cyclic
  Maya diagrams belonging to one of the following signature classes:
  \[ (5), (3,1,1), (1,3,1), (1,1,3), (1,1,1,1,1).\]
  With the normalization $\bpi_4=0$ and $\mu_0=0$, each rational
  solution may be uniquely labeled by one of the above signatures,
  a 4-tuple of arbitrary non-negative integers $(n_1, n_2, n_3, n_4)$, and  a
  permutation $(\pi_0, \pi_1,\pi_2, \pi_3)$ of $(1,2,3,4)$.  For
  each of the above signatures, the corresponding $k$-block coordinates
  of the initial $5$-cyclic Maya diagram are then given by
  \begin{align*}
&k=1 & &(5)& &(0,n_1,n_1+n_2, n_1+n_2+n_3, n_1+n_2+n_3+n_4) \\
&k=3 & &(3,1,1) && (0, n_1 , n_1+ n_2 | n_3 | n_4)\\
&k=3 & &(1,3,1) & &(0| n_1, n_1+n_2, n_1+n_2+n_3 | n_4)\\
&k=3 & &(1,1,3) & &(0| n_1 | n_2, n_2+n_3, n_2+n_3+n_4 )\\
&k=5 & &(1,1,1,1,1) && (0| n_1 | n_2| n_3| n_4)\\
  \end{align*}
\end{thm}

We show specific examples  with shifts $k=1,3$ and $5$ and signatures $(5)$, $(1,1,3)$ and $(1,1,1,1,1)$.

\begin{exercise}\label{ex:51}
\begin{shaded}
Construct a $(5,1)$-cyclic Maya diagram in the signature class $(5)$ with $(n_1,n_2,n_3,n_4)=(2,3,1,1)$ and permutation   $(34210)$. Build the corresponding set of rational solutions to $A_4$-Painlev\'e.
\end{shaded}
The first Maya diagram in the cycle is $M_0=\Xi(0,2,5,6,7)$, depicted in the first row of Figure~\ref{fig:51cyclic}.  The canonical flip sequence is  $\bmu=(0,2,5,6,7)$.  The permutation
  $(34210)$  gives the chain of Maya diagrams shown in Figure
  \ref{fig:51cyclic}.  Note that the permutation specifies the sequence of block coordinates that get shifted by one at each step of the cycle. This type of solutions with signature $(5)$ were already studied in \cite{filipuk2008symmetric}, and they are based on a genus 2 generalization of the generalized Hermite polynomials that appear in the solution of \PIV ($A_2$-Painlev\'e).
\begin{figure}[ht]
  \centering
\begin{tikzpicture}[scale=0.5]

  \path [fill] (0.5,6.5) ++(3,0)
  circle (5pt) ++(1,0)
  circle (5pt) ++ (1,0) circle (5pt)++ (2,0)
  circle (5pt) ++ (3,0)  node[anchor=west] {
    $M_0=\Xi\,(0,2,5,6,7)$};

  \path [fill] (0.5,5.5) ++(3,0) circle (5pt) ++(1,0) circle (5pt) ++
  (1,0) circle (5pt) ++ (5,0)  node[anchor=west] {
    $M_1=\Xi\,(0,2,5,7,7)$};

  \path [fill] (0.5,4.5) ++(3,0)
  circle (5pt) ++(1,0)
  circle (5pt) ++ (1,0) circle (5pt)++ (3,0)
  circle (5pt) ++ (2,0)  node[anchor=west] {
    $M_2=\Xi\,(0,2,5,7,8)$};

  \path [fill] (0.5,3.5) ++(3,0)
  circle (5pt) ++(1,0)
  circle (5pt) ++ (1,0) circle (5pt) ++(1,0) circle (5pt) ++ (2,0)
  circle (5pt) ++ (2,0)  node[anchor=west] {
    $M_3=\Xi\,(0,2,6,7,8)$};

  \path [fill] (0.5,2.5)  ++(4,0)
  circle (5pt) ++ (1,0) circle (5pt) ++(1,0) circle (5pt) ++ (2,0)
  circle (5pt) ++ (2,0)  node[anchor=west] {
    $M_4=\Xi\,(0,3,6,7,8)$};

  \path [fill] (0.5,1.5)  ++(1,0) circle (5pt) ++(3,0)
  circle (5pt) ++ (1,0) circle (5pt) ++(1,0) circle (5pt) ++ (2,0)
  circle (5pt) ++ (2,0)  node[anchor=west] {
    $M_5=\Xi\,(1,3,6,7,8)=M_0+1$};

  \path [fill] (0.5,1.5) ++(0,0) circle (5pt)
  ++(0,1) circle (5pt)
  ++(0,1) circle (5pt)
  ++(0,1) circle (5pt)
  ++(0,1) circle (5pt)
  ++(0,1) circle (5pt);

  \draw  (0,1) grid +(10 ,6);
  \draw[line width=2pt] (1,1) -- ++ (0,6);
  \draw[line width=2pt] (9,1) -- ++ (0,6);

  \foreach \x in {-1,...,8} \draw (\x+1.5,0.5)  node {$\x$};

\end{tikzpicture}
  
\caption{A Maya $5$-cycle with shift $k=1$ for the choice $(n_1,n_2,n_3,n_4)=(2,3,1,1)$ and permutation $\bpi=(34210)$. }
  \label{fig:51cyclic}
\end{figure}
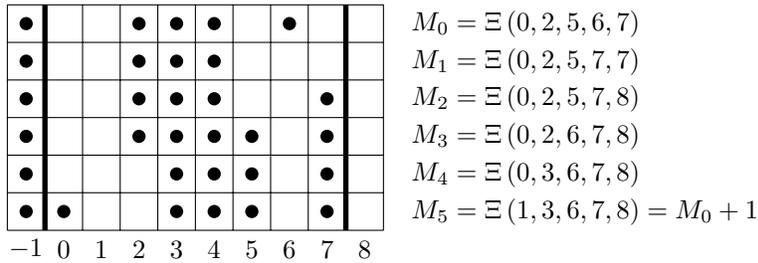

We shall now provide the explicit construction of the rational solution to the $A_4$-Painlev\'e system \eqref{eq:A4system}, by using Proposition~\ref{prop:Mwcorrespondence} and Proposition~\ref{prop:wtof}. The permutation $\bpi=(34210)$ on the canonical sequence $\bmu=(0,2,5,6,7)$ produces the flip sequence $\bmu_\bpi=(6,7,5,2,0)$, so that the values of the $\alpha_i$ parameters given by \eqref{eq:mu2alpha} become $(\alpha_0,\alpha_1,\alpha_2,\alpha_3,\alpha_4)=(-2,4,6,4,-14)$. The pseudo-Wronskians corresponding to each Maya diagram in the cycle are ordinary Wronskians, which will always be the case with the normalization imposed in Remark~\ref{rem:normalization}. They read (see Figure~\ref{fig:51cyclic}):
\begin{eqnarray*}
H_{M_0}(z)&=&\Wr(H_2,H_3,H_4,H_6)\\
H_{M_1}(z)&=&\Wr(H_2,H_3,H_4)\\
H_{M_2}(z)&=&\Wr(H_2,H_3,H_4,H_7)\\
H_{M_3}(z)&=&\Wr(H_2,H_3,H_4,H_5,H_7)\\
H_{M_4}(z)&=&\Wr(H_3,H_4,H_5,H_7)\\
\end{eqnarray*}
where $H_n=H_n(z)$ is the $n$-th Hermite polynomial. The rational solution to the dressing chain is given by the tuple $(w_0,w_1,w_2,w_3,w_4|\alpha_0,\alpha_1,\alpha_2,\alpha_3,\alpha_4)$, where $\alpha_i$ and $w_i$ are given by \eqref{eq:HM2w}-\eqref{eq:mu2alpha} as:
\begin{align*}
w_0(z)&=z+\ddz\Big[\log H_{M_1}(z) - \log H_{M_0}(z)\Big],&& a_0=-2,\\
w_1(z)&=-z+ \ddz\Big[\log H_{M_2}(z)-\log H_{M_1}(z)\Big],&& a_1=4,\\
w_2(z)&=-z+\ddz\Big[\log H_{M_3}(z)- \log H_{M_2}(z)\Big],&& a_2=6,\\
w_3(z)&=z+ \ddz\Big[\log H_{M_4}(z) - \log H_{M_3}(z)\Big],&& a_3=4,\\
w_4(z)&=-z+ \ddz\Big[\log H_{M_0}(z) - \log H_{M_4}(z)\Big],&& a_4=-14.
\end{align*}
Finally, Proposition~\ref{prop:Mwcorrespondence} implies that the corresponding rational solution to the $A_4$-Painlev\'e system \eqref{eq:Ansystem} is given by the tuple $(f_0,f_1,f_2,f_3,f_4|\alpha_0,\alpha_1,\alpha_2,\alpha_3,\alpha_4)$, where
\begin{align*}
f_0(z)&= \ddz\Big[\log H_{M_2}(c_1 z)- \log H_{M_0}(c_1 z)\Big],&& \alpha_0=1,\\
f_1(z)&=z+\ddz\Big[\log H_{M_3}(\cc{1}z) -\log H_{M_1}(\cc{1}z)\Big],&& \alpha_1=-2,\\
f_2(z)&= \ddz\Big[\log H_{M_4}(\cc{1}z) - \log H_{M_2}(\cc{1}z)\Big],&& \alpha_2=-3,\\
f_3(z)&=\ddz\Big[\log H_{M_0}(\cc{1}z) - \log H_{M_3}(\cc{1}z)\Big],&& \alpha_3=-2,\\
f_4(z)&=\ddz\Big[\log H_{M_1}(\cc{1}z) - \log H_{M_4}(\cc{1}z)\Big],&& \alpha_4=7.
\end{align*}
with $\cc{1}^2=-\frac{1}{2}$.

\end{exercise}

\begin{example}
 To illustrate the existence of degenerate Maya cycles, we construct one such degenerate example belonging to the $(5)$ signature class, by choosing 
 $(n_1,n_2,n_3,n_4)=(1,1,2,0)$. The presence of $n_4=0$ means that the first Maya diagram has genus 1 instead of the generic genus 2, with block coordinates given by $M_0=\Xi\,(0,1,2,4,4)$. The canonical flip sequence  $\bmu=(0,1,2,4,4)$ contains two flips at the same site, so it is not unique. Choosing the permutation $(42130)$ produces the chain of Maya  diagrams shown in Figure \ref{fig:51cyclicdegen}.  The explicit construction of the rational solutions follows the same steps as in the previous example, and we shall omit it here. It is worth noting, however, that due to the degenerate character of the chain, three linear combinations of $f_0,\dots,f_4$ will provide a solution to the lower rank $A_2$-Painlev\'e. If the two flips at the same site are performed consecutively in the cycle, the embedding of $A_2^{(1)}$ into $A_4^{(1)}$ is trivial and corresponds to setting two consecutive $f_i$ to zero. This is not the case in this example, as the flip sequence is $\bmu_\bpi=(4,2,1,4,0)$, which produces a non-trivial embedding.
\begin{figure}[ht]
  \centering
\begin{tikzpicture}[scale=0.5]

  \path [fill] (0.5,6.5) 
  ++(2,0)circle (5pt)
  ++(5,0)  node[anchor=west] {  $M_0=\Xi\,(0,1,2,4,4)$};

  \path [fill] (0.5,5.5) 
  ++(2,0) circle (5pt) 
  ++(3,0) circle (5pt)
  ++ (2,0)  node[anchor=west] {$M_1=\Xi\,(0,1,2,4,5)$};

  \path [fill] (0.5,4.5) 
  ++(2,0) circle (5pt)
  ++(1,0) circle (5pt)
  ++(2,0) circle (5pt)
  ++(2,0)  node[anchor=west] {  $M_2=\Xi\,(0,1,3,4,5)$};

  \path [fill] (0.5,3.5) 
  ++(3,0) circle (5pt)
  ++(2,0) circle (5pt)
  ++(2,0)  node[anchor=west] {  $M_3=\Xi\,(0,2,3,4,5)$};

  \path [fill] (0.5,2.5) 
  ++(3,0) circle (5pt)
  ++(4,0)  node[anchor=west] {  $M_4=\Xi\,(0,2,3,5,5)$};

  \path [fill] (0.5,1.5) 
  ++(1,0) circle (5pt)
  ++(2,0) circle (5pt)
  ++(4,0)  node[anchor=west] {  $M_5=\Xi\,(1,2,3,5,5)=M_0+1$};

  \path [fill] (0.5,1.5) ++(0,0) circle (5pt)
  ++(0,1) circle (5pt)
  ++(0,1) circle (5pt)
  ++(0,1) circle (5pt)
  ++(0,1) circle (5pt)
  ++(0,1) circle (5pt);

  \draw  (0,1) grid +(7 ,6);
  \draw[line width=2pt] (1,1) -- ++ (0,6);
  \draw[line width=2pt] (6,1) -- ++ (0,6);

  \foreach \x in {-1,...,5} \draw (\x+1.5,0.5)  node {$\x$};

\end{tikzpicture}
  
  \caption{A degenerate  Maya $5$-cycle with $k=1$ for the choice  $(n_1,n_2,n_3,n_4)=(1,1,2,0)$ and permutation $\bpi=(42130)$.}
  \label{fig:51cyclicdegen}
\end{figure}
\end{example}

\begin{exercise}
\begin{shaded}
Construct a $(5,3)$-cyclic Maya diagram in the signature class $(1,1,3)$ with $(n_1,n_2,n_3,n_4)=(3,1,1,2)$ and permutation $(41230)$. Show the explicit form of the rational solutions to the dressing chain and $A_4$-Painlev\'e.
\end{shaded}
The first Maya diagram has $3$-block coordinates $(0|3|1,2,4)$ and the canonical flip sequence is given by $\bmu=\Theta\,(0|3|1,2,4)=(0, {\color{red}10},{\color{blue} 5,8,14})$.
The permutation $(41230)$ gives the chain of Maya diagrams shown in Figure
  \ref{fig:53cyclic}. Note that, as in Example~\ref{ex:51}, the permutation specifies the order in which the $3$-block coordinates are shifted by +1 in the subsequent steps of the cycle. This type of solutions in the signature class $(1,1,3)$ were not given in \cite{filipuk2008symmetric}, and they are new to the best of our knowledge.

\begin{figure}[ht]
  \centering
\begin{tikzpicture}[scale=0.5]
  \path [fill,red] (0.5,6.5)  
  ++(2,0) circle (5pt)
  ++(3,0) circle (5pt)
  ++(3,0) circle (5pt) ;

  \path [fill,blue] (0.5,6.5)  
  ++(3,0) circle (5pt)
  ++(6,0) circle (5pt) 
  ++(3,0) circle (5pt) ;

  \path (17.5,6.5)  node[anchor=west] {$M_0=\Xi_3\,(0|3|1,2,4)$};

  \path [fill,red] (0.5,5.5)  
  ++(2,0) circle (5pt)
  ++(3,0) circle (5pt)
  ++(3,0) circle (5pt) ;

  \path [fill,blue] (0.5,5.5)  
  ++(3,0) circle (5pt)
  ++(6,0) circle (5pt) 
  ++(3,0) circle (5pt) 
  ++(3,0) circle (5pt) ;
  \path (17.5,5.5)  node[anchor=west] {$M_1=\Xi_3\,(0|3|1,2,5)$};

  \path [fill,red] (0.5,4.5)  
  ++(2,0) circle (5pt)
  ++(3,0) circle (5pt)
  ++(3,0) circle (5pt)
  ++(3,0) circle (5pt) ;

  \path [fill,blue] (0.5,4.5)  
  ++(3,0) circle (5pt)
  ++(6,0) circle (5pt) 
  ++(3,0) circle (5pt) 
  ++(3,0) circle (5pt) ;
  \path (17.5,4.5)  node[anchor=west] {$M_2=\Xi_3\,(0|4|1,2,5)$};

  \path [fill,red] (0.5,3.5)  
  ++(2,0) circle (5pt)
  ++(3,0) circle (5pt)
  ++(3,0) circle (5pt)
  ++(3,0) circle (5pt) ;

  \path [fill,blue] (0.5,3.5)  
  ++(3,0) circle (5pt)
  ++(3,0) circle (5pt)
  ++(3,0) circle (5pt) 
  ++(3,0) circle (5pt) 
  ++(3,0) circle (5pt) ;
  \path (17.5,3.5)  node[anchor=west] {$M_3=\Xi_3\,(0|4|2,2,5)$};

  \path [fill,red] (0.5,2.5)  
  ++(2,0) circle (5pt)
  ++(3,0) circle (5pt)
  ++(3,0) circle (5pt)
  ++(3,0) circle (5pt) ;
  \path [fill,blue] (0.5,2.5)  
  ++(3,0) circle (5pt)
  ++(3,0) circle (5pt)
  ++(6,0) circle (5pt) 
  ++(3,0) circle (5pt) ;
  \path (17.5,2.5)  node[anchor=west] {$M_4=\Xi_3\,(0|4|2,3,5)$};

  \path [fill,black] (0.5,1.5)  
  ++(1,0) circle (5pt);  
  \path [fill,red] (0.5,1.5)  
  ++(2,0) circle (5pt)
  ++(3,0) circle (5pt)
  ++(3,0) circle (5pt)
  ++(3,0) circle (5pt) ;
  \path [fill,blue] (0.5,1.5)  
  ++(3,0) circle (5pt)
  ++(3,0) circle (5pt)
  ++(6,0) circle (5pt) 
  ++(3,0) circle (5pt) ;
  \path (17.5,1.5)  node[anchor=west] {$M_5=\Xi_3\,(1|4|2,3,5)=M_0+3$};

  \path [fill,blue] (0.5,1.5) ++(0,0) circle (5pt)
  ++(0,1) circle (5pt)
  ++(0,1) circle (5pt)
  ++(0,1) circle (5pt)
  ++(0,1) circle (5pt)
  ++(0,1) circle (5pt);

  \path [fill,red] (-.5,1.5) ++(0,0) circle (5pt)
  ++(0,1) circle (5pt)
  ++(0,1) circle (5pt)
  ++(0,1) circle (5pt)
  ++(0,1) circle (5pt)
  ++(0,1) circle (5pt);

  \path [fill,black] (-1.5,1.5) ++(0,0) circle (5pt)
  ++(0,1) circle (5pt)
  ++(0,1) circle (5pt)
  ++(0,1) circle (5pt)
  ++(0,1) circle (5pt)
  ++(0,1) circle (5pt);

  \draw  (-2,1) grid +(19 ,6);
  \draw[line width=2pt] (1,1) -- ++ (0,6);
  \draw[line width=2pt] (16,1) -- ++ (0,6);

  \foreach \x in {0,...,15} \draw (\x+1.5,0.5)  node {$\x$};
\end{tikzpicture}
  \caption{A Maya $5$-cycle with shift $k=3$ for the choice $(n_1,n_2,n_3,n_4)=(3,1,1,2)$ and permutation $\bpi=(41230)$. }
  \label{fig:53cyclic}
\end{figure}

We proceed to build the explicit rational solution to the $A_4$-Painlev\'e system \eqref{eq:A4system}. In this case, the permutation $\bpi=(41230)$ on the canonical sequence $\bmu=(0,10,5,8,14)$ produces the flip sequence $\bmu_\bpi=(14,10,5,8,0)$, so that the values of the $\alpha_i$ parameters given by \eqref{eq:mu2alpha} become $(\alpha_0,\alpha_1,\alpha_2,\alpha_3,\alpha_4)=(8,10,-6,16,-34)$. The pseudo-Wronskians corresponding to each Maya diagram in the cycle are ordinary Wronskians, which will always be the case with the normalization imposed in Remark~\ref{rem:normalization}. They read (see Figure~\ref{fig:53cyclic}):
\begin{eqnarray*}
H_{M_0}(z)&=&\Wr(H_1,H_2,H_4,H_7,H_8,H_{11})\\
H_{M_1}(z)&=&\Wr(H_1,H_2,H_4,H_7,H_8,H_{11},H_{14})\\
H_{M_2}(z)&=&\Wr(H_1,H_2,H_4,H_7,H_8,H_{10},H_{11},H_{14})\\
H_{M_3}(z)&=&\Wr(H_1,H_2,H_4,H_5,H_7,H_8,H_{10},H_{11},H_{14})\\
H_{M_4}(z)&=&\Wr(H_1,H_2,H_4,H_5,H_7,H_{10},H_{11},H_{14})\\
\end{eqnarray*}
where $H_n=H_n(z)$ is the $n$-th Hermite polynomial. The rational solution to the dressing chain is given by the tuple $(w_0,w_1,w_2,w_3,w_4|\alpha_0,\alpha_1,\alpha_2,\alpha_3,\alpha_4)$, where $\alpha_i$ and $w_i$ are given by \eqref{eq:HM2w}-\eqref{eq:mu2alpha} as:
\begin{align*}
w_0(z)&=-z+\ddz\Big[\log H_{M_1}(z) - \log H_{M_0}(z)\Big],&& a_0=8,\\
w_1(z)&=-z+ \ddz\Big[\log H_{M_2}(z)-\log H_{M_1}(z)\Big],&& a_1=10,\\
w_2(z)&=-z+ \ddz\Big[\log H_{M_3}(z) - \log H_{M_2}(z)\Big],&& a_2=-6,\\
w_3(z)&=z+ \ddz\Big[\log H_{M_4}(z) - \log H_{M_3}(z)\Big],&& a_3=16,\\
w_4(z)&=-z+ \ddz\Big[\log H_{M_0}(z) - \log H_{M_4}(z)\Big],&& a_4=-34.
\end{align*}
Finally, Proposition~\ref{prop:wtof} implies that the corresponding rational solution to the $A_4$-\p\ system \eqref{eq:Ansystem} is given by the tuple $(f_0,f_1,f_2,f_3,f_4|\a_0,\a_1,\a_2,\a_3,\a_4)$, where
\begin{align*}
f_0(z)&=\tfrac13z+\ddz\Big[\log H_{M_2}(\cc{2}z) - \log H_{M_0}(\cc{2}z)\Big],&& \alpha_0=-\tfrac43,\\
f_1(z)&=\tfrac13z+\ddz\Big[\log H_{M_3}(\cc{2}z) -\log H_{M_1}(\cc{2}z)\Big],&& \alpha_1=-\tfrac53,\\
f_2(z)&= \ddz\Big[\log H_{M_4}(\cc{2}z) - \log H_{M_2}(\cc{2}z)\Big],&& \alpha_2=1,\\
f_3(z)&=\ddz\Big[\log H_{M_0}(\cc{2}z) -\log H_{M_3}(\cc{2}z)\Big],&& \alpha_3=-\tfrac83,\\
f_4(z)&=\tfrac13z+\ddz\Big[\log H_{M_1}(\cc{2}z) - \log H_{M_4}(\cc{2}z)\Big],&& \alpha_4=\tfrac{17}{3}.
\end{align*}
with $\cc{2}^2=-\tfrac16$.

\end{exercise}

\begin{exercise}
\begin{shaded}
Construct a $(5,5)$-cyclic Maya diagram in the signature class $(1,1,1,1,1)$ with $(n_1,n_2,n_3,n_4)=(2,3,0,1)$ and permutation $(32410)$. Show the explicit form of the rational solutions to the dressing chain and $A_4$-Painlev\'e.
\end{shaded}
With the above choice, the first Maya diagram o the cycle has $5$-block coordinates $(0|2|3|0|1)$, and the canonical flip sequence is given by  $\bmu=\Theta\,(0|2|3|0|1)=(0,{\color{red}11},{\color{blue}17},{\color{brown}3},{\color{green}
    9})$. The permutation $(32410)$ gives the chain of Maya diagrams shown in Figure \ref{fig:55cyclic}.  Note that, as it happens in the previous examples, the permutation specifies the order in which the $5$-block coordinates are shifted by +1 in the subsequent steps of the cycle. This type of solutions with signature $(1,1,1,1,1)$ were already studied in \cite{filipuk2008symmetric}, and they are based on a generalization of the Okamoto polynomials that appear in the solution of \PIV ($A_2$-Painlev\'e).
    
\begin{figure}[ht]
  \centering
\includegraphics[width=1.2\textwidth]{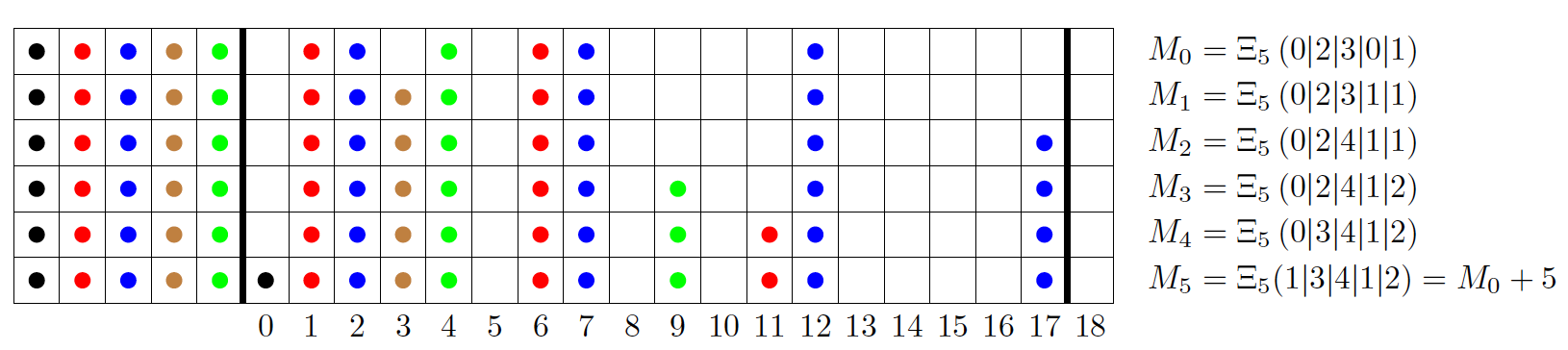}
  \caption{A Maya $5$-cycle with shift $k=5$ for the choice $(n_1,n_2,n_3,n_4)=(2,3,0,1)$ and permutation $\bpi=(32410)$.}
  \label{fig:55cyclic}
\end{figure}

We proceed to build the explicit rational solution to the $A_4$-Painlev\'e system \eqref{eq:A4system}. In this case, the permutation $\bpi=(32410)$ on the canonical sequence $\bmu=(0,11,17,3,9)$ produces the flip sequence $\bmu_\bpi=(3,17,9,11,0)$, so that the values of the $\alpha_i$ parameters given by \eqref{eq:mu2alpha} become $(\alpha_0,\alpha_1,\alpha_2,\alpha_3,\alpha_4)=(-28,16,-4,22,-16)$. The pseudo-Wronskians corresponding to each Maya diagram in the cycle are ordinary Wronskians, which will always be the case with the normalization imposed in Remark~\ref{rem:normalization}. They read:
\begin{eqnarray*}
H_{M_0}(z)&=&\Wr(H_1,H_2,H_4,H_6,H_7,H_{12})\\
H_{M_1}(z)&=&\Wr(H_1,H_2,H_3,H_4,H_6,H_7,H_{12})\\
H_{M_2}(z)&=&\Wr(H_1,H_2,H_3,H_4,H_6,H_7,H_{12},H_{17})\\
H_{M_3}(z)&=&\Wr(H_1,H_2,H_3,H_4,H_6,H_7,H_9,H_{12},H_{17})\\
H_{M_4}(z)&=&\Wr(H_1,H_2,H_3,H_4,H_6,H_7,H_9,H_{11},H_{12},H_{17})\\
\end{eqnarray*}
where $H_n=H_n(z)$ is the $n$-th Hermite polynomial. The rational solution to the dressing chain is given by the tuple $(w_0,w_1,w_2,w_3,w_4|\alpha_0,\alpha_1,\alpha_2,\alpha_3,\alpha_4)$, where $\alpha_i$ and $w_i$ are given by \eqref{eq:HM2w}-\eqref{eq:mu2alpha} as:
\begin{align*}
w_0(z)&=-z+\ddz\Big[\log H_{M_1}(z) - \log H_{M_0}(z)\Big],&& a_0=-28\\
w_1(z)&=-z+ \ddz\Big[\log H_{M_2}(z) - \log H_{M_1}(z)\Big],&& a_1=16,\\
w_2(z)&=-z+ \ddz\Big[\log H_{M_3}(z) - \log H_{M_2}(z)\Big],&& a_2=-4,\\
w_3(z)&=-z+ \ddz\Big[\log H_{M_4}(z) - \log H_{M_3}(z)\Big],&& a_3=22\\
w_4(z)&=-z+ \ddz\Big[\log H_{M_0}(z) - \log H_{M_4}(z)\Big],&& a_4=-16.
\end{align*}
Finally, Proposition~\ref{prop:Mwcorrespondence} implies that the corresponding rational solution to the $A_4$-\p\ system \eqref{eq:Ansystem} is given by the tuple $(f_0,f_1,f_2,f_3,f_4|\a_0,\a_1,\a_2,\a_3,\a_4)$, where
\begin{align*}
f_0(z)&=\tfrac15z+\ddz\Big[\log H_{M_2}(\cc{3}z) - \log H_{M_0}(\cc{3}z)\Big],&& \alpha_0=\tfrac{14}{5},\\
f_1(z)&=\tfrac15z+\ddz\Big[\log H_{M_3}(\cc{3}z) -\log H_{M_1}(\cc{3}z)\Big],&& \alpha_1=-\tfrac85,\\
f_2(z)&=\tfrac15z+ \ddz\Big[\log H_{M_4}(\cc{3}z)-\log H_{M_2}(\cc{3}z)\Big],&& \alpha_2=\tfrac25,\\
f_3(z)&=\tfrac15z+\ddz\Big[\log H_{M_0}(\cc{3}z) - \log H_{M_3}(\cc{3}z)\Big],&& \alpha_3=-\tfrac{11}5,\\
f_4(z)&=\tfrac15z+\ddz\Big[\log H_{M_1}(\cc{3}z) - \log H_{M_4}(\cc{3}z)\Big],&& \alpha_4=\tfrac85.
\end{align*}
with $\cc{3}^2=-\tfrac1{10}$.

\end{exercise}

\section*{Acknowledgements}
The research of DGU has been supported in part by the Spanish Ministerio de Ciencia, Innovaci\'on y Universidades and ERDF under contracts RTI2018-100754-B-I00  and  PGC2018-096504-B-C33.  The research of RM was supported in part by NSERC grant
RGPIN-228057-2009. DGU would like to thank Mama Foupagnigni, Wolfram Koepf, the Volkswagen Stiftung and the African Institute of Mathematical Sciences for their hospitality during the Workshop on Introduction to Orthogonal Polynomials and Applications, Duala (Cameroon), where these lectures were first taught.


%

\providecommand{\bysame}{\leavevmode\hbox to3em{\hrulefill}\thinspace}
\providecommand{\MR}{\relax\ifhmode\unskip\space\fi MR }
\providecommand{\MRhref}[2]{%
  \href{http://www.ams.org/mathscinet-getitem?mr=#1}{#2}
}
\providecommand{\href}[2]{#2}

\end{document}